\documentclass[accepted]{uai2026} 
                        
\usepackage[american]{babel}

\usepackage{natbib} 
\usepackage{mathtools} 
\usepackage{booktabs} 
\usepackage{tikz} 
\usepackage{graphicx}
\usepackage{xcolor}
\usepackage{amsmath, amsthm, amssymb}
\usepackage{tikz-qtree}
\usepackage[dvipsnames]{xcolor}
\usepackage{mathtools}
\usepackage[utf8]{inputenc} 
\usepackage[T1]{fontenc}    
\usepackage{hyperref}       
\usepackage{url}            
\usepackage{booktabs}       
\usepackage{algorithm}%
\usepackage{algorithmicx}%
\usepackage{algpseudocode}%
\usepackage{amsfonts}       
\usepackage{nicefrac}       
\usepackage{microtype}      
\usepackage{xcolor}         
\usepackage{graphicx}
\usepackage{adjustbox}
\usepackage{amssymb, bm}
\usepackage{stmaryrd}
\usepackage{bbm}
\usepackage{float}
\usepackage{soul}
\usepackage{thm-restate}
\usepackage{tcolorbox}
\tcbuselibrary{theorems,skins}
\usepackage{enumitem}

\usepackage{mathrsfs}
\usepackage{xcolor}

\renewcommand{\matrix}[1]{\mathbf{#1}}

\newcommand{\lmap}[1]{\mathbf{#1}}

\newcommand{\ip}[2]{\langle #1, #2\rangle}

\newcommand{\nnz}[1]{\operatorname{nnz}\left(#1\right)}

\newcommand{\Var}[1]{\mathrm{Var}\left(#1\right)}

\renewcommand{\vec}[1]{\mathbf{#1}}
\newcommand{\E}{\mathbb{E}}
\newcommand{\Expt}[1]{\E \left[ #1 \right]}
\newcommand{\Exptnb}[1]{\E [ #1 ]}
\newcommand{\indic}[1]{\mathbf{1}\left[#1\right]}

\renewcommand{\Re}[1]{\mathrm{Re}\left\{#1\right\}}
\renewcommand{\Im}[1]{\mathrm{Im}\left\{#1\right\}}
\newcommand{\Renb}[1]{\mathrm{Re}\{#1\}}
\newcommand{\Imnb}[1]{\mathrm{Im}\{#1\}}

\newcommand{\norm}[1]{\left\|#1 \right\|_2}
\newcommand{\normSq}[1]{\left\|#1 \right\|_2^2}
\newcommand{\Frob}[1]{\left\|#1 \right\|_F}

\newcommand{\modulus}[1]{\left|#1 \right|}
\newcommand{\modulusSq}[1]{\modulus{#1}^2}
\newcommand{\modulusnb}[1]{|#1 |}
\newcommand{\modulusSqnb}[1]{\modulusnb{#1}^2}

\newcommand{\CtR}{\mathrm{CtR}}

\newcommand{\kCtR}[2]{\widehat{k}_{\CtR}\left(#1,#2\right)}
\newcommand{\kCtRxy}{\kCtR{\vec{x}}{\vec{y}}}
\newcommand{\kC}[2]{\widehat{k}_{\mathrm{C}}\left(#1,#2\right)}
\newcommand{\kCxy}{\kC{\vec{x}}{\vec{y}}}
\newcommand{\kxy}{\widehat{k}\left(\vec{x},\vec{y}\right)}
\newcommand{\phiC}{\Phi_{\mathrm{C}}}

\newcommand{\phiCx}{\Phi_{\mathrm{C}}(\vec{x}^{\otimes p})}
\newcommand{\phiCy}{\Phi_{\mathrm{C}}(\vec{y}^{\otimes p})}
\newcommand{\phiCX}{\Phi_{\mathrm{C}}(X)}
\newcommand{\phiCY}{\Phi_{\mathrm{C}}(Y)}

\newcommand{\conj}[1]{\overline{#1}}

\renewcommand{\Pr}{\text{Pr}}
\newcommand{\R}{\mathbb{R}}
\newcommand{\Z}{\mathbb{Z}}

\newtheorem{thm}{Theorem}

  \newtheorem{remark}[thm]{Remark}

\newtheorem{definition}{Definition}

\newtcbtheorem[
  number within=section
]{boxeddefinition}{Definition}{
  colback=white,
  colframe=black,
  left=2pt,right=2pt,
  top=2pt,bottom=2pt,
  boxsep=1pt,
  colbacktitle=gray!8,
  coltitle=black,
  fonttitle=\bfseries,
  enhanced,
  boxrule=0.8pt
}{def}

\newif\ifwideappendix

\title{Improving TensorSketch Using Complex Random Variables}

\author[1]{\href{mailto:<cs24resch02002@iith.ac.in>?Subject=REG: Improving TensorSketch Using Complex Random Variables paper}{Amit Sharma}\textsuperscript{*}}
\author[1]{\href{mailto:<cs24mtech12006@iith.ac.in>?Subject=REG: Improving TensorSketch Using Complex Random Variables paper}{Mohammad Azhar Khan}{}\textsuperscript{*}}
\author[1]{\href{mailto:<rameshwar@cse.iith.ac.in>?Subject=REG: Improving TensorSketch Using Complex Random Variables paper}{Rameshwar Pratap}}
\author[2]{\href{mailto:<kk054@bucknell.edu>?Subject=REG: Improving TensorSketch Using Complex Random Variables paper}{Keegan Kang}}
\affil[1]{%
    Department of Computer Science and Engineering\\
    IIT Hyderabad\\
    India
}
\affil[2]{%
    Department of Mathematics and Statistics\\
    Bucknell University\\
    USA
}
  
\begin{document}
\maketitle
{\let\thefootnote\relax\footnotetext{\textsuperscript{*} Equal contribution.}}

\begin{abstract}
\texttt{TensorSketch} by~\cite{pham2013fast,kar2012random} provides efficient sketching algorithms for high-dimensional polynomial kernels $\vec{x}^{\otimes p} \in \R^{d^p}$. \cite{kar2012random} uses dense Johnson-Lindenstrauss (JL)-type projections with computational cost $O(pDd)$ , where $D$ denotes the sketch dimension, whereas~\cite{pham2013fast} extends the sparse \texttt{CountSketch}~\citep{count_sketch} algorithm, yielding a faster algorithm for high-dimensional sparse inputs with running time $O\big(p(\nnz{\vec{x}} + D \log D)\big)$. However, the variance of both estimators grows exponentially with the polynomial degree $p$, scaling as $3^{p}/D$. Recent work by~\cite{pmlr-v206-wacker23a} showed that using complex-valued distribution reduces this dependence to $2^{p}/D$ for the approach of~\cite{kar2012random}. However, their method relies on dense JL-type projections with computational cost $O(pDd)$ and does not extend to the algorithm of~\cite{pham2013fast}.

In this work, we introduce a simple variant of \texttt{TensorSketch}~\citep{pham2013fast} that achieves the same variance bound as~\cite{pmlr-v206-wacker23a}, while retaining its advantage of the input-sparsity running time. We validate our results with supporting experiments on synthetic and real-world datasets.

\end{abstract}

\section{Introduction}\label{sec:intro}

Polynomial kernels (Definition~\ref{def:poly_kernel}) 
are widely used in machine learning to model higher-order, non-linear interactions among input features.
For vectors $\vec{x}, \vec{y} \in \mathbb{R}^d$, a degree-$p$ polynomial kernel is defined as
\(
k(\vec{x},\vec{y}) = \ip{\vec{x}}{\vec{y}}^p
\).
This kernel is equivalent to mapping $\vec{x} \in \mathbb{R}^d$ to its $p$-fold Kronecker product
\(
\vec{x}^{\otimes p} \in \mathbb{R}^{d^p},
\)
which is
\(
\ip{\vec{x}^{\otimes p}}{\vec{y}^{\otimes p}}
\equiv \ip{\vec{x}}{\vec{y}}^p .
\)

As the dimension $d^p$ grows exponentially with $p$, directly computing these inner products is computationally infeasible. A dense Johnson-Lindenstrauss (JL) transform~\citep{johnson1984extensions} can reduce their dimensionality while approximately preserving their pairwise inner product, however applying it requires time proportional to $d^p$, which is exponential in $p$.

To address this inefficiency, prior works such as \cite{kar2012random} and \cite{pham2013fast} proposed randomized sketching techniques to approximate polynomial kernels efficiently.
\cite{kar2012random} introduced random feature maps based on JL-type projections. 
They define a randomized linear map 
$\lmap{S} : \mathbb{R}^{d^p} \to \mathbb{R}^{D}$ by
$\lmap{S}(\vec{x}^{\otimes p}) := \left(\matrix{W}_1 \vec{x} \,\odot\, \cdots \,\odot\, \matrix{W}_p \vec{x}\right)/\sqrt{D}$,
where each $\matrix{W}_i \in \mathbb{R}^{D \times d}$ is a random projection matrix
(e.g., Gaussian or Rademacher with \textit{i.i.d.} entries), and $\odot$ denotes the
element-wise (Hadamard) product.
Such that
\(
\kxy :=
\ip{\lmap{S}(\vec{x}^{\otimes p})}{\lmap{S}(\vec{y}^{\otimes p})}.
\) This estimator can be termed as a JL-type variant of \texttt{TensorSketch}, and  its computational cost scales with $O(p D d)$. 
 \cite{pham2013fast} propose \texttt{TensorSketch}, which combines \texttt{CountSketch}~\citep{count_sketch} with Fast Fourier Transform (FFT)-based convolution to compute the sketch implicitly. 
This avoids forming the full vector and runs in input-sparsity time 
\(
O\big(p(\nnz{\vec{x}} + D \log D)\big), 
\) making it preferred over \cite{kar2012random} for high-dimensional sparse data. The variance of both algorithms  \cite{kar2012random, pham2013fast,pham2025tensorsketchfastscalable} grows as $3^p/D$.

Recent progress by \cite{wacker2022improved} introduced a complex-valued variant of JL-type \texttt{TensorSketch}. It improves the variance dependence of the sketch on the polynomial degree, reducing it from $3^{p}$ to $2^{p}$. 
This improvement is achieved by incorporating complex-valued randomness into the sketch. 
However, the resulting sketch vectors are
complex-valued and cannot be directly
compared with their real-valued
counterparts. To address this, \cite{pmlr-v206-wacker23a} proposed the \emph{Complex-to-Real (CtR)} construction. 
It first computes a complex random feature map of the embedding dimension half of the size of its real counterpart and then forms a real embedding by concatenating its real and imaginary parts.
This preserves inner products and achieves improved variance bounds while yielding a real-valued sketch.
However, as a JL-type method, it still requires dense multiplications with cost $O(p D d)$, which can be inefficient for high-dimensional sparse data.  

\noindent In this work, we address the above limitation by developing the
\emph{Complex-to-Real (CtR)} variant of \texttt{TensorSketch}. 
The proposed construction obtain variance improvements comparable to the 
complex JL-type estimator of~\cite{pmlr-v206-wacker23a}, while retaining the 
input-sparsity running time guarantees of \texttt{TensorSketch}~\citep{pham2013fast}.

\noindent\textbf{Contributions.}
We propose a \emph{Complex-to-Real} variant of \texttt{TensorSketch} for degree-$p$ polynomial kernels. It combines a random function whose values are drawn independently and uniformly from the fourth roots of unity with FFT-based tensor sketching while producing real-valued embeddings. We prove that the resulting estimator is unbiased for $\ip{\vec{x}^{\otimes p}}{\vec{y}^{\otimes p}}$. We also derive an upper bound on variance with $2^{p}$-type dependence on the degree, improving over the classical \texttt{TensorSketch}~\citep{pham2013fast} variance scaling, while retaining the sketching time $O(p(\nnz{\vec{x}} + D \log D))$. The algorithm is defined in Definition~\ref{def:ctr-tensorsketch}, and its theoretical guarantees are stated in Theorem~\ref{thm:complex_countsketch_degree_p}.

 We note that the use of complex-valued random variables in sketching algorithms has been explored in prior work. For example, \cite{pmlr-v206-wacker23a,wacker2022improved} employ complex random variables to construct sketches for polynomial kernels, while \cite{meyer2026hutchinson} use them for trace estimation of implicit matrices. The key idea in \cite{pmlr-v206-wacker23a,wacker2022improved} is to sketch $\mathbf{x}^{\otimes p}$ using element-wise products of independent sketches of the factors $(\mathbf{x}_1,\ldots,\mathbf{x}_p)$. Owing to this independence structure, Khintchine's inequality (Definition~\ref{def:Khintchine Inequality}) can be applied, and a simple induction on $p$ yields a tighter upper bound in the complex setting than in the real-valued case.

In contrast, \texttt{TensorSketch} is built upon the hash-based \texttt{CountSketch} framework, where the sketch components are not independent and no direct analogue of Khintchine's inequality is available. As a result, the techniques from prior work do not extend to our setting. Moreover, our analysis of \texttt{CountSketch} with complex-valued random variables (Theorem~\ref{thm:cs-unbiasedness}, Appendix) shows that, by itself, the complex construction of \texttt{CountSketch} provides no variance reduction over its real-valued version. Therefore, our result demonstrating improved variance for \texttt{TensorSketch} with complex random variables is non-trivial and stems from the vanishing of certain cross terms in the variance analysis, which does not occur in the real-valued setting.
\noindent \textbf{Implications of our results.}
\texttt{TensorSketch}~\citep{pham2013fast} enables efficient training of linear SVMs~\citep{sun2018but,li2019towards}, supports scalable deep learning and the theoretical analysis of over-parameterized neural networks~\citep{yehudai2019power,zandieh2021scaling}, and has been successfully applied to compact bilinear pooling for fine-grained visual recognition~\citep{gao2016compact} as well as multimodal fusion architectures~\citep{fukui2016multimodal}. Polynomial kernels themselves are widely adopted in applications including natural language processing~\citep{goldberg2008splitting}, recommender systems~\citep{rendle2010factorization}, and genomics~\citep{aschard2016incorporating}. Our proposed CtR \texttt{TensorSketch} achieves more accurate estimates than \texttt{TensorSketch}~\citep{pham2013fast} while maintaining the same asymptotic time complexity, making it a promising alternative for the above applications.

\noindent \textbf{Organization of the paper.}
The remainder of this paper is organized as follows. 
Section~\ref{sec:related_work} reviews prior work on polynomial kernel approximation and randomized sketching. 
Section~\ref{sec:preliminaries} presents necessary background, including \texttt{CountSketch}, \texttt{TensorSketch}, polynomial kernels, and the CtR framework. 
Section~\ref{sec:ctr_sketches}, proposes our algorithm and states theoretical guarantees on its accuracy and efficiency.   
Section~\ref{Experiments_sec} reports empirical comparisons of real and JL-type methods with our approach. Section~\ref{sec:conclusion} concludes and outlines some future research directions.
\section{Related Work}\label{sec:related_work}

Polynomial kernels can be expressed via tensor feature maps, where the representation is the tensor product $\bigotimes_{i=1}^{p} \vec{x}_i$ of vectors $\vec{x}_1 \in \mathbb{R}^{d_1},\ldots,\vec{x}_p \in \mathbb{R}^{d_p}$ to capture higher–order interactions. Explicit construction requires $O\!\left(\prod_{i=1}^{p} d_i\right)$ memory and is infeasible even for moderate $p$ or dimensions $d$. To avoid this blow-up, prior work proposes sketching methods that compute $\matrix{S}\!\left(\bigotimes_{i=1}^{p} \vec{x}_i\right)$ without forming the full tensor, including JL-type product embeddings~\citep{kar2012random} and hashing-based \texttt{TensorSketch} constructions built on \texttt{CountSketch}~\citep{pham2013fast,pham2025tensorsketchfastscalable}. For $\vec{x},\vec{y}\in\mathbb{R}^{d}$, both    \texttt{TensorSketch} estimators have variance
\(
\le\frac{3^p - 1}{D}\!\left(\normSq{\vec{x}} \normSq{\vec{y}}\right),
\)
which yields $3^{p}$-type variance growth. Both JL-type and hashing-based estimators scale as $\Theta(3^{p}/D)$ for degree-$p$ features, but \texttt{TensorSketch}~\citep{pham2013fast,pham2025tensorsketchfastscalable} is more efficient due its input-sparsity time complexity.

Recent work shows that complex-valued distributions reduce the variance of randomized feature maps. Complex JL-type constructions for polynomial kernels achieve lower variance than their real-valued counterparts~\citep{wacker2022improved}. In particular, replacing real Rademacher or Gaussian variables with complex-valued distributions improves the variance dependence from $3^{p}$-type to $2^{p}$-type. Further work by~\cite{pmlr-v206-wacker23a} introduced the \emph{Complex-to-Real (CtR)} construction, preserving the variance improvement of complex-valued sketches while yielding real-valued embeddings. These embeddings enables direct comparison with real sketches. However, these approaches rely on dense JL-style projections and therefore incur higher computational cost especially for high-dimensional sparse data.

Building on the hashing-based sketching framework of~\citep{pham2013fast}, we introduce a \emph{Complex-to-Real (CtR)} variant of \texttt{TensorSketch} for polynomial kernel approximation. Our method uses fourth roots-of-unity  instead of real Rademacher variables and applies a structured \emph{complex-to-real} transformation, yielding unbiased real-valued embeddings with improved variance dependence.

In contrast, our method achieves variance reduction through a lightweight modification of the sketch construction. This is followed by a \emph{Complex-to-Real (CtR)} conversion to produce a real-valued embedding that makes it possible to compare it fairly with its real counterpart. The resulting sketch preserves the original structure and input-sparsity running time while achieving provable variance improvements over original \texttt{TensorSketch}~\citep{pham2013fast}.

Variance reduction for randomized sketching has been widely studied using statistical variance reduction techniques such as control variates (CV) and maximum likelihood estimation (MLE). The CV method reduces variance by leveraging a correlated auxiliary variable with a known expectation, while MLE estimates unknown parameters by maximizing the likelihood of the observed data, often yielding statistically efficient estimators. These techniques have been extensively investigated for a variety of classical randomized sketching methods, including Johnson--Lindenstrauss (JL) transforms~\cite{Li2006very}, CountSketch~\cite{DBLP:conf/uai/PratapK21}, Tug-of-War sketches~\cite{DBLP:conf/acda/PratapVK21}, signed random projections~\cite{Pmlr-v80-kang18b}, feature hashing~\cite{DBLP:journals/ml/VermaPT22}, and compressed matrix multiplication~\cite{DBLP:journals/ipl/VermaDPT25}, among others. However, to the best of our knowledge, analogous variance reduction techniques have not yet been developed for \texttt{TensorSketch}. While these approaches can substantially reduce estimator variance, they typically incur additional computational and algorithmic overhead.

\section{Preliminaries}\label{sec:preliminaries}
\textbf{Notation:}
We use the following notation in the paper. Bold lowercase letters
(e.g., $\vec{x},\vec{y}$) denote vectors, bold uppercase letters
(e.g., $\matrix{S},\matrix{T}$) denote matrices. For a positive integer $D$, we write
$[D]:=\{1,2,\ldots,D\}$ for the corresponding index set. For $d,p \in \mathbb{N}$, we denote by $[d]^p$ the set of all
$p$-tuples $(i_1,\ldots,i_p)$ with $i_j \in [d]$. The sets
$\R$, $\Z$ and $\mathbb{C}$ denote the real, integer and complex domains, respectively.
For any vector $\vec{x}$,
$\nnz{\vec{x}}$ denotes the number of its nonzero entries.
The symbol $\ip{\vec{x}}{\vec{y}}$ denotes the standard inner
product, $\norm{\vec{x}}$ the Euclidean norm, and $\Frob{\matrix{S}}$ the
Frobenius norm of a matrix. For a complex number $z \in \mathbb{C}$ written as $z = a + ib$, its complex conjugate is defined as
$\conj{z}$. The norm of $z$ is
\(
|z| = \sqrt{a^2 + b^2} = \sqrt{z\,\conj{z}}.
\) The Kronecker product is written as $\otimes$, while “$\cdot$’’
denotes standard matrix multiplication and $\odot$ denotes the elementwise (Hadamard) product.
The indicator function is denoted by $\indic{\cdot}$.
The probabilistic quantities use $\Pr[\cdot]$ for probability,
$\Expt{\cdot}$ for expectation and $\Var{\cdot}$ for variance.

\begin{definition}[\texttt{CountSketch}~\citep{count_sketch}] \label{count_sketch_def}
Given an input vector $\vec{y} \in \R^{d}$, the \texttt{CountSketch} is a randomized linear map $\mathbf{T} \in \R^{D \times d}$ that maps $\vec{y}$ to a lower-dimensional vector $\vec{z} = \matrix{T}\vec{y} \in \R^D.$ The \texttt{CountSketch} matrix $\matrix{T}$ is constructed by two hash functions: (a) $h \colon [d] \to [D]$ a $2$-wise independent hash function,  and (b) $s: [d] \to \{1,-1\}$ a  $4$-wise independent random sign function.

The $j^{th}$ entry of  vector $\vec{z} \in \R^D$ is computed as,
$
z_j = \sum_{h(i) = j} s(i)\, y_i, \; \forall j \in [D] .
$
\noindent The time complexity of computing the \texttt{CountSketch} is 
$O\!\left(\nnz{\vec{y}}\right)$.

For pairwise inner--product analysis, let $\vec{x},\vec{y} \in \R^{d}$ and define the 
\texttt{CountSketch} inner--product estimator as
\[
\kxy := \ip{\matrix{T}\vec{x}}{\matrix{T}\vec{y}} .
\]
Then the estimator is unbiased
\[
\Expt{\kxy }
=
\ip{\vec{x}}{\vec{y}},
\]
and its variance satisfies
\[
\Var{\kxy}
\le
\frac{1}{D}
\Big(
\ip{\vec{x}}{\vec{y}}^{2}
+
\normSq{\vec{x}} \normSq{\vec{y}}
-
2 \sum_{i} x_i^2 y_i^2
\Big).
\]

\end{definition}

\begin{definition}[\texttt{TensorSketch} of Degree $p$~\citep{pham2013fast,pham2025tensorsketchfastscalable}]
\label{def:tensorsketch_degree_p}
Let $p \ge 2$. For each $t \in [p]$, let 
$h_t : [d] \to [D]$ be $2$-wise independent hash functions and 
$\sigma_t : [d] \to \{-1,+1\}$ be $4$-wise independent random sign functions. 
The degree-$p$ \texttt{TensorSketch} matrix 
$\matrix{S} \in \R^{D \times d^p}$ is defined for all 
$r \in [D]$ and $(i_1,\dots,i_p) \in [d]^p$ by
\[
\matrix{S}_{r,(i_1,\dots,i_p)} 
= \left(\prod_{t=1}^{p} \sigma_t(i_t)\right)
\indic{
\sum_{t=1}^{p} h_t(i_t) \equiv r \pmod D
}.
\]

\noindent
For any $\vec{x} \in \mathbb{R}^d$, the sketch 
$\matrix{S}\!\left(\vec{x}^{\otimes p}\right)$ can be computed implicitly in 
time $O(p(\nnz{\vec{x}} + D \log D))$ using FFT-based 
convolution, without explicitly forming $\vec{x}^{\otimes p}$.\\
Let $\vec{x},\vec{y}\in\R^{d}$, then estimate for pairwise inner--product is defined as follows: 
\[
\kxy
:=
\ip{
\matrix{S}\!\left(\vec{x}^{\otimes p}\right)
}{
\matrix{S}\!\left(\vec{y}^{\otimes p}\right)
}.
\]
Then \texttt{TensorSketch} provides an unbiased estimator of the 
degree--$p$ polynomial kernel
\[
\Expt{\kxy}
=
\ip{\vec{x}}{\vec{y}}^{p},
\]
and its variance satisfies
\begin{align*}
\Var{\kxy}
&\le
\frac{1}{D}
\Bigg(
\Big(
2\ip{\vec{x}}{\vec{y}}^{2}
+
\normSq{\vec{x}}\,\normSq{\vec{y}} \cdots \\
&\qquad
\cdots - 2\sum_{i=1}^{d} x_i^{2}y_i^{2}
\Big)^{p}
-
\ip{\vec{x}}{\vec{y}}^{2p}
\Bigg), \\
&\le
\frac{3^{p}-1}{D}
\,\norm{\vec{x}}^{2p}
\,\norm{\vec{y}}^{2p}.
\end{align*}
\end{definition}

\begin{definition}[Polynomial Kernel~\citep{10.7551/mitpress/4175.001.0001}]\label{def:poly_kernel}
We consider polynomial kernels of degree $p\in\mathbb{N}$ of the form
\[
k(\vec{x},\vec{y})=\left(\gamma\,\vec{x}^\top \vec{y}+\nu\right)^p,
\]
for $\vec{x},\vec{y}\in \R^d$ with $\gamma,\nu\ge 0$. The parameters $\gamma$ and $\nu$
can be absorbed into the inputs by introducing the augmented vectors
\[
\tilde{\vec{x}}=\bigl(\sqrt{\gamma}\,\vec{x}^\top,\;\sqrt{\nu}\bigr)^\top,\qquad
\tilde{\vec{y}}=\bigl(\sqrt{\gamma}\,\vec{y}^\top,\;\sqrt{\nu}\bigr)^\top \in \R^{d+1}.
\]
With this transformation, the kernel admits a homogeneous representation,
\[
\bigl(\gamma\,\vec{x}^\top \vec{y}+\nu\bigr)^p
=
(\tilde{\vec{x}}^\top \tilde{\vec{y}})^p
=
\bigl(\tilde{\vec{x}}^{\otimes p}\bigr)^\top \tilde{\vec{y}}^{\otimes p},
\]
\noindent where $\tilde{\vec{x}}^{\otimes p}$ denotes the $p$-fold tensor product of $\tilde{\vec{x}}$.
Hence, without loss of generality, we may treat the polynomial kernel as a homogeneous
kernel in the augmented space $\R^{d+1}$.
\end{definition}

\begin{definition}[Khintchine Inequality~\cite{haagerup2007best,pmlr-v206-wacker23a}]\label{def:Khintchine Inequality}
Let $\mathbf{x}=(x_1,\ldots,x_d)\in\mathbb{R}^d$ and let
$\{\varepsilon_i\}_{i=1}^d$ be independent Rademacher random variables.
For every $p>0$, there exists a constant $K_p$ such that
\begin{align}
\left(
\mathbb{E}
\left|
\sum_{i=1}^{d} x_i \varepsilon_i
\right|^p
\right)^{1/p}
\leq
K_p \|\mathbf{x}\|_2.
\end{align}

The optimal constants are known explicitly.

\begin{itemize}
    \item If $\varepsilon_i\in\{-1,+1\}$ are real-valued Rademacher variables, then
    \begin{align}
    K_p =
    \begin{cases}
    1, & 0<p\le 2,\\[2mm]
    \sqrt{2}\,
    \pi^{-\frac{1}{2p}}
    \Gamma\!\left(\frac{p+1}{2}\right)^{\frac{1}{p}},
    & p>2.
    \end{cases}
    \end{align}

    \item If $\varepsilon_i$ are complex Rademacher variables taking values
    in $\{1,-1,i,-i\}$ uniformly at random, then
    \begin{align}
    \hat{K_p}
    =
    \Gamma\!\left(\frac{p}{2}+1\right)^{\frac{1}{p}}.
    \end{align}
\end{itemize}

Moreover, for every $p>2$, the complex constant $\hat{K_p}$ is strictly smaller than
the corresponding real constant $K_p$.
\end{definition}

\begin{definition}[Framework for Complex-to-Real (CtR) Sketches]\label{CtR_Framework}
\cite{pmlr-v206-wacker23a} suggests a framework of analysing sketching algorithms involving complex random variables.   Let $z = a + ib$ be a complex random variable  with $a,b \in \R$, we have $\modulusSq{z} = a^2 + b^2
\quad\text{and}\quad \Re{z^2} = a^2 - b^2.$
Combining both gives $a^2 = \frac{1}{2}\left(\modulusSq{z} + \Re{z^2}\right).$
The scalar $a$ is real-valued and its variance
$\Var{a} = \Expt{a^2} - \Expt{a}^2$
is therefore
\begin{align}
\Var{a}
& = \frac{1}{2}\Renb{\Exptnb{\modulusSqnb{z}} + \Exptnb{z^{2}} - 2\Exptnb{a}^{2}}. \label{eq:variance of real part of estimate}
\end{align}
Let $\vec{x},\vec{y} \in \mathbb{R}^{d}$ and let 
$\phiC:\mathbb{R}^{d^{p}} \rightarrow \mathbb{C}^{\frac{D}{2}}$. 
Then $\kCxy := \phiCx\conj{\phiCy}^{\top} \in \mathbb{C}$ 
is a complex-valued estimate of the polynomial kernel $k(\mathbf{x,y})$, and hence 
$\kCxy = \Renb{\kCxy} + i  \Imnb{\kCxy}$. 
The \emph{Complex-to-Real} sketch is defined as follows:
\begin{align*}
   &\kCtRxy := \Renb{\kCxy},\\
&=\!\Re{\phiCx}^\top \! \Re{ \phiCy} \!+\\ &\quad \cdots +\! \Im{\phiCx}^\top\! \Im{\phiCy},\\ 
&=\Phi_{\CtR}(\vec{x})^\top \Phi_{\CtR}(\vec{y}).
\end{align*}
where, $\Phi_{\CtR}(\vec{x}) : = [\Re{\phiC},\Im{\phiC}] \in \mathbb{R}^{D}.$ We now derive the variance of $\kCtRxy$ using Equation~\eqref{eq:variance of real part of estimate}:
\begin{align}
&\Var{\kCtRxy}\notag\\
 &= \frac{1}{2}\Renb{\Exptnb{\modulusSqnb{\widehat{k}_{C}}} + \Exptnb{\widehat{k}_{C}^{2}} + 2\Exptnb{\Renb{\widehat{k}_{C}}}^{2}} \notag\\
 &= \frac{1}{2}\Renb{\Exptnb{\modulusSqnb{\widehat{k}_{C}}} + \Exptnb{\widehat{k}_{C}^{2}} + 2\Exptnb{\widehat{k}_{C}}^{2}}, \label{eq:ctr_variance_structure}
\end{align}
where, $\widehat{k}_{C}\!:=\!\kCxy\! \in\! \mathbb{C}.$ This framework plays a key role in computing the expectation and variance of $\kCtRxy$. 
\end{definition}

\begin{definition}[\emph{Complex-to-Real} Polynomial Sketch~\citep{pmlr-v206-wacker23a}]
Let $p \in \mathbb{N}$ and $D = 2k$ for some $k \in \mathbb{N}$. 
For each $i \in [p]$, let 
$\matrix{W}_i \in \mathbb{C}^{D/2 \times d}$ be a random matrix whose rows are sampled 
independently from a zero-mean distribution satisfying
$\Expt{\mathbf{w} \mathbf{w}^{*}} = \matrix{I}_d$ (e.g., complex Gaussian or complex Rademacher) where $*$ refers to conjugate-transpose. Define the complex random feature map
\[
\Phi_C(\vec{x}^{\otimes p}) 
:= \sqrt{\frac{2}{D}}\;
\big( \matrix{W_1} \vec{x} \odot \matrix{W_2} \vec{x} \odot \cdots \odot \matrix{W_p} \vec{x} \big)
\;\in\; \mathbb{C}^{D/2},
\]

Then, the \emph{Complex-to-Real (CtR) sketch} of $\vec{x} \in \R^{d^p}$ is the 
real-valued vector $\Phi_{\CtR}(\vec{x}^{\otimes p}) \in \mathbb{R}^{D}$ defined by
\[
\Phi_{\CtR}(\vec{x}^{\otimes p})
:=
\begin{bmatrix}
\Re{\Phi_{\CtR}(\vec{x}^{\otimes p})} \\
\Im{\Phi_{\CtR}(\vec{x}^{\otimes p})}
\end{bmatrix} \in \mathbb{R}^{D}.
\]

For any $\vec{x}^{\otimes p},\vec{y}^{\otimes p} \in \R^{d^p}$, the kernel estimator can be defined as
\begin{align*}
    \kCtRxy &=\Phi_{\CtR}(\vec{x})^\top \Phi_{\CtR}(\vec{y}).
\end{align*}

For both Gaussian and Rademacher constructions,
the CtR estimator satisfied the following:

\begin{align*}
\Expt{\kCtRxy}&=\ip{\vec{x}}{\vec{y}}^{p},\\
\Var{\kCtRxy}
&\le
\frac{2^{p+1}-2}{D}
\,\|\vec{x}\|_2^{2p}
\,\|\vec{y}\|_2^{2p}.
\end{align*}
\end{definition}

\section{\emph{Complex-to-Real(CtR)} \texttt{TensorSketch}}
\label{sec:ctr_sketches}

\begin{table*}[t]
\centering
\begin{tabular}{|c|c|c|}
\hline
\textbf{Algorithm} & \textbf{Variance} & \textbf{Sketching Time} \\
\hline
\texttt{TensorSketch}(CtR) & $\frac{2^{p+1} - 2}{D}\norm{\vec{x}}^{2p}\,\norm{\vec{y}}^{2p}$ & $O(p(\nnz{\vec{x}} + \nnz{\vec{y}} + D \log D))$ \\
\hline
\texttt{TensorSketch}(Real)\citep{pham2013fast,pham2025tensorsketchfastscalable} & $\frac{3^{p} - 1}{D}\norm{\vec{x}}^{2p}\,\norm{\vec{y}}^{2p}$ & $O(p(\nnz{\vec{x}} + \nnz{\vec{y}} + D \log D)) $ \\
\hline
JL(CtR Radamacher)\citep{pmlr-v206-wacker23a} &   $\frac{2^{p+1} - 2}{D}\norm{\vec{x}}^{2p}\,\norm{\vec{y}}^{2p}$ & $O(pDd)$ \\
\hline
JL(CtR Guassian)\citep{pmlr-v206-wacker23a} &  $\frac{2^{p+1} - 2}{D}\norm{\vec{x}}^{2p}\,\norm{\vec{y}}^{2p}$ & $O(pDd)$ \\
\hline
\end{tabular}
\caption{\textbf{Comparison of variance bounds and sketching time for CtR \texttt{TensorSketch} and baseline methods. \label{tab:ctr sketches}}
Here, $\vec{x},\vec{y}\in\mathbb{R}^{d}$ are input vectors, $p$ denotes the polynomial degree, and $D$ is the sketch dimension. All methods provide unbiased estimates of the inner product between the vectors $\vec{x}^{\otimes p}, \vec{y}^{\otimes p} \in \mathbb{R}^{d^{p}}$, i.e.,  $\ip{ \vec{x}^{\otimes p}}{\vec{y}^{\otimes p} }$.}
\end{table*}
\noindent
In this section, we first define Complex \texttt{CountSketch} (Definition~\ref{def:ctr-countsketch}). Building on its circular convolution structure, we then introduce \emph{Complex-to-Real} \texttt{TensorSketch}. (Definition~\ref{def:ctr-tensorsketch}) along with the associated kernel. Theorem~\ref{thm:complex_countsketch_degree_p} discusses the guarantees of the estimators proposed in Definition~\ref{def:ctr-tensorsketch}.

\begin{definition}[Complex \texttt{CountSketch} Mapping]
\label{def:ctr-countsketch}
Let $\vec{x} \in \mathbb{R}^d$ and $D$ be the   sketching dimension.   
Sample a Complex \texttt{CountSketch} matrix $\matrix{C} \in \mathbb{C}^{D \times d}$ with the following two independent functions
\begin{itemize}
    \item $h : [d] \to [D]$  is an universal hash function that assigns each coordinate
independently and uniformly to one of the $D$ buckets, and
    \item $s : [d] \to \{1, \omega, \omega^{2}, \omega^{3}\}$ is a random function whose values are drawn
independently and uniformly from the four fourth roots of unity.
\end{itemize}
\noindent
The $j$-th coordinate of the sketched vector $\matrix{C}\vec{x}$ is given by
\begin{align*}
(\matrix{C}\vec{x})_{j}
= \sum_{i=1}^{d} s(i)\,  \indic{h(i)=j} x_i\,,
\qquad \forall\, j \in [D].
\end{align*}

\end{definition}
We define  CtR \texttt{TensorSketch}, which extends the Complex \texttt{CountSketch} by applying $p$ independent sketches and combining them using FFT-based convolution to efficiently sketch the vector $\vec{x}^{\otimes p},\vec{y}^{\otimes p} \in \mathbb{R}^{d^{p}}$.
\begin{definition}{(\emph{Complex-to-Real (CtR)} \texttt{TensorSketch} Mapping)}
\label{def:ctr-tensorsketch}
Let $\vec{x} \in \mathbb{R}^{d}$ and fix an integer $p \ge 1$.  
Define a Complex \texttt{TensorSketch} mapping
\(
\matrix{C} : \mathbb{R}^{d^p} \to \mathbb{C}^{D/2}
\)
constructed from $p$ independent Complex \texttt{CountSketch} mappings.  
For each $r \in [p]$, let
\begin{itemize}
    \item $h_r : [d] \to [D/2]$ be an universal hash function that assigns each coordinate
independently and uniformly to one of the $D$ buckets, and
    \item $s_r : [d] \to \{1, \omega, \omega^{2}, \omega^{3}\}$ be a random function whose values are drawn
independently and uniformly from the fourth roots of unity.
\end{itemize}

\noindent
For each $r \in [p]$, let $\matrix{C}_r \in \mathbb{C}^{D/2 \times d}$ denote the Complex \texttt{CountSketch} matrix induced by functions $(h_r, s_r)$ as in Definition~\ref{def:ctr-countsketch}.  
The \texttt{TensorSketch} of $\vec{x}^{\otimes p} \in \mathbb{R}^{d^p}$ is defined implicitly via convolution of the $p$ sketches, and can be computed efficiently using the Fast Fourier Transform as
\begin{align}
\phiC(\vec{x}^{\otimes p}) &:= \matrix{C} \vec{x}^{\otimes p},\\
&= \operatorname{FFT}^{-1}\!\left(
\bigodot_{r=1}^{p}
\operatorname{FFT}\!\left(\matrix{C}_r \vec{x}\right)
\right)
\in
\mathbb{C}^{D/2}, \label{eq:complex-tensorsketch}
\end{align}
where the product is taken element-wise. We define the CtR sketch as 
\begin{align*}
\!\Phi_{\CtR}(\vec{x}^{\otimes p}) \!
:=\!
\Bigl(
\Renb{\phiC(\vec{x}^{\otimes p})_{1}}, \ldots,
\Renb{\phiC(\vec{x}^{\otimes p})_{D/2}},\\
\Imnb{\phiC(\vec{x}^{\otimes p})_{1}}, \ldots,
\Imnb{\phiC(\vec{x}^{\otimes p})_{D/2}}
\Bigr)^{\top}
\in \mathbb{R}^{D}.
\end{align*}

\end{definition}
A detailed proof of results (Lemma~\ref{lem:sign_fun_4th_root}, \ref{lem:complex_to_real AMS}, and Theorem~\ref{thm:complex_countsketch_degree_p}) stated in this section is presented in Appendix~\ref{CtR_tensorsketch_appendix}.

\begin{restatable}{lem}{fourththroot}
\label{lem:sign_fun_4th_root}
Let $\omega = e^{2\pi i / 4} = i$, and let $s : [d] \to \{1,\ \omega,\ \omega^{2},\ \omega^{3}\}$
be a random  function such that the values $\{s(i)\}_{i \in [d]}$ are drawn independently and uniformly from the four fourth roots of unity. Then, for every $i \in [d]$, the following identities hold
\begin{align*}
    \Expt{s(i)} &= 0, \\
    \Exptnb{\modulusSq{s(i)}} &= 1, \\
    \Expt{s(i)^{2}} &= 0.
\end{align*}
Furthermore, for any pair of distinct indices $i \neq j$, independence implies
\begin{align*}
    \Exptnb{s(i)\,\conj{s(j)}} 
        = \Exptnb{s(i)}\,\Expt{\conj{s(j)}}
        = 0.
\end{align*}
\end{restatable}

\noindent
Theorem~\ref{thm:complex_countsketch_degree_p} establishes the unbiasedness, variance bound, and
sketching time of the CtR \texttt{TensorSketch}.
\begin{restatable}{thm}{ctrTensorSketch}
\label{thm:complex_countsketch_degree_p}
Let $\vec{x}, \vec{y} \in \mathbb{R}^{d}$ and
$\vec{x}^{\otimes p}, \vec{y}^{\otimes p} \in \mathbb{R}^{d^{p}}$. Let $\Phi_{\CtR}$ denote a \emph{Complex-to-Real} \texttt{TensorSketch} as stated in Definition~\ref{def:ctr-tensorsketch}.  
Let $\kCtRxy := \Phi_{\CtR}(\vec{x}^{\otimes p})^{\top}\Phi_{\CtR}(\vec{y}^{\otimes p})$, then $\kCtRxy$ satisfies
\begin{align*}
\Expt{
\kCtRxy
}
&= \ip{\vec{x}^{\otimes p}}{ \vec{y}^{\otimes p}}
= \ip{\vec{x}}{\vec{y}}^p,\\
\Var{
    \kCtRxy
    }
&\le \frac{2^{p + 1}-2}{D}\,
\norm{\vec{x}}^{2p}\, \norm{\vec{y}}^{2p}.
\end{align*}

\noindent Moreover, the sketches $\Phi_{\CtR}(\vec{x}^{\otimes p})$ and
$\Phi_{\CtR}(\vec{y}^{\otimes p})$ can be computed in $O(p(\nnz{\vec{x}} + D \log D)) \ \text{and} \, O(p(\nnz{\vec{y}} + D \log D))$
time, respectively.
\end{restatable}

\begin{proof}
According to Definition~\ref{CtR_Framework}, we have 
$\kCtRxy := \Re{\kCxy}$, and therefore we first analyze $\kCxy$. 
Recall that $\kCxy := \phiCx \overline{\phiCy}^{\top}$, where 
$\phiCx := \matrix{C}\vec{x}^{\otimes p}$ and 
$\phiCy := \matrix{C}\vec{y}^{\otimes p}$ according to 
Definition~\ref{def:ctr-tensorsketch}. 
In particular, $\matrix{C}\vec{x}^{\otimes p}, 
\matrix{C}\vec{y}^{\otimes p} \in \mathbb{C}^{D/2}$ denote the complex 
\texttt{CountSketch} of the vectors 
$\vec{x}^{\otimes p}, \vec{y}^{\otimes p} \in \mathbb{R}^{d^{p}}$, 
constructed using the derived functions 
$H : [d]^p \mapsto [D/2]$ and 
$S : [d]^p \to \{1,\omega,\omega^{2},\omega^{3}\}$ defined as
\begin{align*}
H(i_1,\ldots,i_p) &= \left( \sum_{j=1}^{p} h_j(i_j) \right) \bmod D/2, \\
S(i_1,\ldots,i_p) &= \prod_{j=1}^{p} s_j(i_j).
\end{align*}

\noindent In the proof, we denote $X := \vec{x}^{\otimes p}$ and $Y := \vec{y}^{\otimes p}$. 
Let $u,v \in [d]^p$ be the multi-indices corresponding to the entries of 
$X,Y \in \mathbb{R}^{d^p}$. For a multi-index $u=(j_1,\dots,j_p)$, the entry 
of $X$ is defined as $X_u = \prod_{k=1}^{p} x_{j_k}$, where the associated 
linearized index is $u = 1 + \sum_{k=1}^{p}(j_k - 1)\prod_{k'=1}^{p-k} d$. 
The entries of $Y$ are defined analogously. 

We begin by analyzing $\kCxy$ as defined below,
\begin{align*}
&\kCxy := \phiCX \overline{\phiCY}^{\top} = \ip{\matrix{C} X}{\conj{\matrix{C}Y}}, \\
  &= \sum_{u,v \in [d]^p}
      X_{u}\, Y_{v}\, 
      S(u)\, \overline{S(v)}\ \indic{H(u) = H(v)}, \\
  &= \ \langle X, Y \rangle
   \!+\!
     \sum_{u \ne v}
      X_{u}\, Y_{v}\,
      S(u)\, \overline{S(v)}\, \indic{H(u) = H(v)}.
\end{align*}
Using Lemma~\ref{lem:sign_fun_4th_root}, we have 
$\Expt{S(u)\,\conj{S(v)}} = 0$ for all $u \neq v$. Hence,
\begin{align*}
\Expt{\kCxy}
= \ip{X}{Y} = \ip{\mathbf{x}^{\otimes p}}{\mathbf{y}^{\otimes p}}
= \ip{\vec{x}}{\vec{y}}^p .
\end{align*}
According to Definition~\ref{CtR_Framework}, we know $\kCtRxy := \Re{\kCxy}$, we have
\begin{align}
\Expt{\kCtRxy} = \Expt{\Re{\kCxy}} = \langle \vec{x}, \vec{y} \rangle^{p}. \label{eq:expectation-tensorsketch}
\end{align}
\noindent Using Equation~\ref{eq:ctr_variance_structure}, the variance of $\kCtRxy$ is
\begin{align} &\Var{\kCtRxy} = \frac{1}{2}\operatorname{Re}\!\Bigg\{ \Expt{ \modulusSq{\kCxy} } + \cdots \nonumber \\ &\quad \cdots + \Expt{ \left( \kCxy \right)^2 } - 2 \left( \Expt{ \kCxy } \right)^2 \Bigg\}. \label{final variance ctr tensorsketch main} \end{align}

For the variance analysis, we need to compute 
$\Expt{\modulusSq{\kCxy}}$ and $\Expt{\left(\kCxy\right)^{2}}$.
By applying Lemma~\ref{lem:sign_fun_4th_root} and 
Lemma~\ref{lem:complex_to_real AMS}, we obtain the following bounds:
\begin{align} &\Expt{\modulusSq{\kCxy}} \le \ip{\vec{x}}{\vec{y}}^{2p} + \frac{2}{D}\Bigg(\Bigg(\ip{\vec{x}}{\vec{y}}^{2} + \notag\\
&\qquad \ \cdots+ \normSq{\vec{x}}\normSq{\vec{y}} - \sum_{i=1}^{d} x_i^{2} y_i^{2}\Bigg)^{p}  - \ip{\vec{x}}{\vec{y}}^{2p}\Bigg),\label{second moment norm ctr main} \end{align}
and
\begin{align} &\Expt{ \left( \kCxy\right)^{2}}\le \ip{\vec{x}}{\vec{y}}^{2p} + \frac{2}{D}\Bigg(\Bigg(2\ip{\vec{x}}{\vec{y}}^{2} - \notag\\
&\qquad \qquad \quad \cdots - \sum_{i=1}^{d} x_i^{2} y_i^{2}\Bigg)^{p} - \ip{\vec{x}}{\vec{y}}^{2p}\Bigg).\label{second moment ctr main} \end{align}
Substituting Equations~\eqref{second moment norm ctr main}, 
\eqref{second moment ctr main}, and \eqref{eq:expectation-tensorsketch} into 
Equation~\eqref{final variance ctr tensorsketch main} yields
\begin{align}
\Var{\kCtRxy}
\le \frac{2^{p+1}-2}{D}\,
\norm{\vec{x}}^{2p}\norm{\vec{y}}^{2p}.
\end{align}

\textbf{Time Complexity.}
To compute the sketch according to Equation~\eqref{eq:complex-tensorsketch}, each complex \texttt{CountSketch} can be computed in $O(\nnz{\vec{x}})$ time. The convolution of the $p$ sketches is implemented using FFT in $O(pD\log D)$ time. Therefore, the overall time complexity of Complex-to-Real (CtR)~\texttt{TensorSketch} applied to $\vec{x}^{\otimes p}$ is $O\!\big(p(\nnz{\vec{x}} + D\log D)\big)$.
\end{proof}
\begin{remark}
    Theorem~\ref{thm:complex_countsketch_degree_p}, exhibits an improved exponential dependence on $p$ in the variance bound of the estimate, decreasing from the $3^p/D$ bound for real \texttt{TensorSketch}~\citep{pham2013fast} to $2^p/D$, while retaining the same input-sparsity sketching time. A tabular comparison among the baselines on the variance bound and running time is presented in Table~\ref{tab:ctr sketches}.
\end{remark}

Bounds derived in the Lemma~\ref{lem:complex_to_real AMS} play a key role in the proof of Theorem~\ref{thm:complex_countsketch_degree_p}. The sketch described in this lemma is a complex-valued variant of the classical AMS sketch~\cite{ALON1999137} for polynomial kernel approximation~\cite{10.5555/1347082.1347163,braverman2010ams4wiseindependenceproduct}. In particular, the variance analysis of CtR~\texttt{TensorSketch} reduces to analyzing products of independently randomized linear sketches. Each bucket of CtR \texttt{TensorSketch} behaves like a product of $p$ independent AMS-type sketches of the form $Z_{s_j}(\vec{x}) = \sum_{i=1}^d x_i s_j(i)$. When expanding the moments $\mathbb{E}[|Z|^2]$ and $\mathbb{E}[Z^2]$ in the proof of Theorem~\ref{thm:complex_countsketch_degree_p}, the resulting expressions factor across the $p$ independent random functions drawn uniformly from fourth roots of unity. Lemma~\ref{lem:complex_to_real AMS} provides an exact evaluation of  moments under fourth roots-of-unity random functions, allowing the full degree-$p$ moment to be written as the $p$-th power of a single-sketch expression.
\begin{restatable}{lem}{AMSSketch}
\label{lem:complex_to_real AMS}
Let $\vec{x}, \vec{y} \in \mathbb{R}^{d}$, let $p \ge 1$ be an integer, and let 
$s_{1}, \ldots, s_{p} : [d] \to \{1, \omega, \omega^{2}, \omega^{3}\}$ 
be independent random functions, each taking values uniformly from the four fourth roots of unity (as in Lemma~\ref{lem:sign_fun_4th_root}).  
Define
\begin{align*}
    Z \;=\; \prod_{j=1}^{p} Z_{s_j}(\vec{x}) \, \conj{Z_{s_j}(\vec{y})},
\end{align*}
where
\begin{align*}
    Z_{s_j}(\vec{x}) &= \sum_{i=1}^{d} x_i\, s_j(i), 
    &
    Z_{s_j}(\vec{y}) &= \sum_{i=1}^{d} y_i\, s_j(i).
\end{align*}
Then,
\begin{align*}
    \Expt{Z} &= \ip{\vec{x}}{\vec{y}}^p, \\
    \Var{Z} &\leq 2^{p}\, \norm{\vec{x}}^{2p}\, \norm{\vec{y}}^{2p}.
\end{align*}
\end{restatable}
\begin{proof}
Let
\[
Z
=
\prod_{j=1}^{p}
Z_{s_j}(\vec{x})\,
\conj{Z_{s_j}(\vec{y})},
\qquad
Z_{s_j}(\vec{x})
=
\sum_{i=1}^{d} x_i s_j(i).
\]

\medskip
\noindent
\textbf{Expectation.}
For a fixed $j$,
\[
\Expt{
Z_{s_j}(\vec{x})
\conj{Z_{s_j}(\vec{y})}
}
=
\sum_{i,k}
x_i y_k\,
\Expt{s_j(i)\conj{s_j(k)}}.
\]
By Lemma~\ref{lem:sign_fun_4th_root},
$\Expt{s_j(i)\overline{s_j(k)}}=0$ for $i\neq k$
and equals $1$ for $i=k$.
Hence,
\[
\Expt{
Z_{s_j}(\vec{x})
\conj{Z_{s_j}(\vec{y})}
}
=
\ip{\vec{x}}{\vec{y}}.
\]
Independence across $j$ then gives
\[
\Expt{Z}
=
\prod_{j=1}^{p}
\ip{\vec{x}}{\vec{y}}
=
\ip{\vec{x}}{\vec{y}}^{p}.
\]

\medskip
\noindent
\textbf{Second moments.}
Since the sign functions are independent across $j$,
\begin{align*}
    \Expt{\modulusSq{Z}}
&=
\prod_{j=1}^{p}
\Expt{\modulusSqnb{Z_{s_j}(\vec{x})
\conj{Z_{s_j}(\vec{y})}}
},\\
\Expt{Z^2}
&=
\prod_{j=1}^{p}
\Expt{
\big(
Z_{s_j}(\vec{x})
\conj{Z_{s_j}(\vec{y})}
\big)^2
}.
\end{align*}

Expanding one factor and using the fourth–moment structure of fourth roots-of-unity random functions
(Lemma~\ref{lem:sign_fun_4th_root}),
only index configurations that form pairs survive. 
The detailed calculation yields moments as follows
\[
\Expt{
\modulusSqnb{Z_{s_j}(\vec{x})
\conj{Z_{s_j}(\vec{y})}}
}
=
\ip{\vec{x}}{\vec{y}}^2
+
\normSq{\vec{x}} \normSq{\vec{y}}
-
\sum_{i=1}^{d} x_i^2 y_i^2,
\]
and
\[
\Expt{
\big(
Z_{s_j}(\vec{x})
\conj{Z_{s_j}(\vec{y})}
\big)^2
}
=
2\ip{\vec{x}}{\vec{y}}^2
-
\sum_{i=1}^{d} x_i^2 y_i^2.
\]

Therefore, taking $p$-times product of the terms give moment bounds for $Z$
\begin{align*}
\Exptnb{\modulusSqnb{Z}}
&=
\Big(
\ip{\vec{x}}{\vec{y}}^2
+
\normSq{\vec{x}} \normSq{\vec{y}}
-
\sum_{i=1}^{d} x_i^2 y_i^2
\Big)^p,\\
\Expt{Z^2}
&=
\Big(
2\ip{\vec{x}}{\vec{y}}^2
-
\sum_{i=1}^{d} x_i^2 y_i^2
\Big)^p.
\end{align*}
\noindent
\textbf{Variance bound.}
Using Equation~\ref{eq:variance of real part of estimate}, and applying the Cauchy-Schwarz inequality to upper bound 
$\ip{\vec{x}}{\vec{y}}^2 \le \normSq{\vec{x}} \normSq{\vec{y}}$, together with the fact that 
$\sum_i x_i^2 y_i^2 \ge 0$, we obtain
\[
\Var{Z}
\le
2^{p}\,
\norm{\vec{x}}^{2p}
\norm{\vec{y}}^{2p}.
\]
\end{proof}
\section{Experiments}\label{Experiments_sec}
\noindent \textbf{Experimental setup.}
\begin{figure*}[t]
    \centering
    \includegraphics[width=0.95\textwidth]{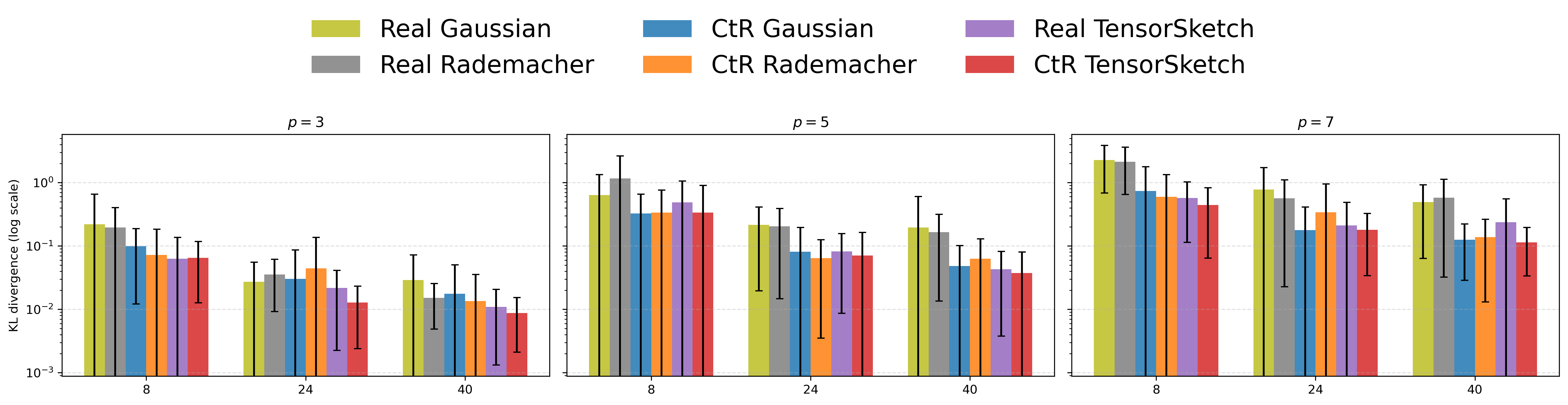}
    \caption{
     \textbf{KL divergence on the COD--RNA dataset.}
    We report KL divergence between the exact degree-$p$ polynomial kernel and the kernel
    reconstructed from sketch features for $p \in \{3,5,7\}$.
    Methods include Real and CtR (\emph{complex-to-real}) Gaussian and Rademacher JL sketches,
    as well as Real and CtR TensorSketch.
    Results are averaged over 20 independent trials.
    Sketch dimension is varied as $D \in \{d,3d,5d\}$.
    }
    \label{fig:codrna_kl_sketch}
\end{figure*}
\begin{figure*}[t]
    \centering
    \includegraphics[width=0.95\textwidth]{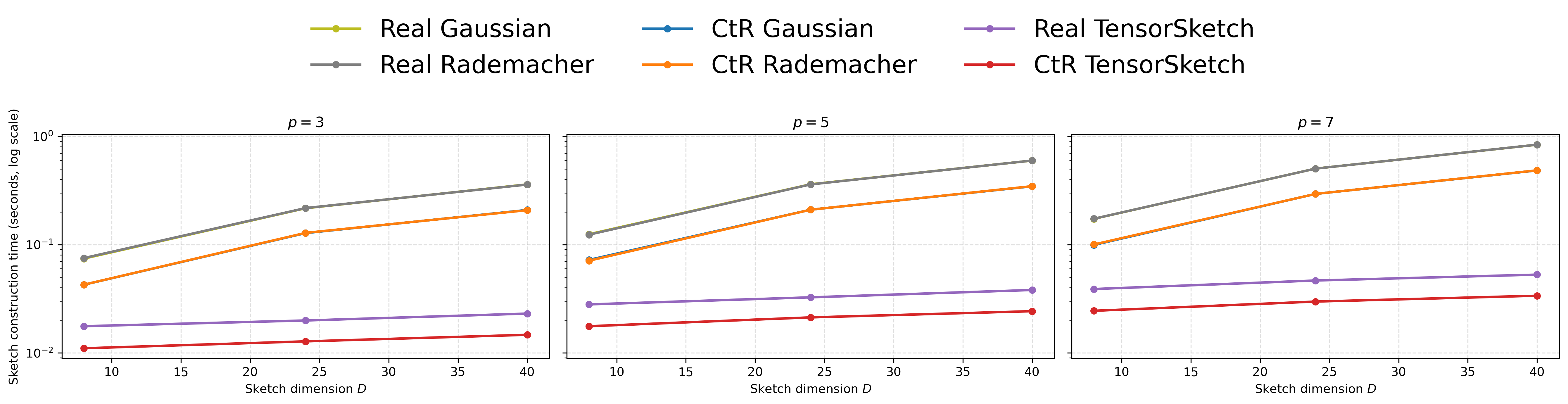}
    \caption{
     \textbf{Wall-clock sketch construction time on the COD--RNA dataset.}
    We compare Real and CtR (\emph{complex-to-real}) Gaussian and Rademacher JL sketches, with Real and CtR TensorSketch for $p \in \{3,7,10\}$.
    The sketch dimension is varied as $D \in \{d,3d,5d\}$.
    Each point reports the average sketch construction time over $20$ independent trials,
    measured on identical normalized input data.
    Here as well all dense
    JL-type sketches exhibit nearly overlapping construction times across
    sketch dimensions $D$ and polynomial degrees $p$.
    }
    \label{fig:codrna_time_sketch}
\end{figure*}
We compare our proposed CtR \texttt{TensorSketch} against standard \texttt{TensorSketch}~\citep{pham2013fast}. Along with this we included JL-type polynomial sketches (Gaussian and Rademacher~\citep{kar2012random}) together with their CtR variants~\citep{pmlr-v206-wacker23a} for comparison. All methods are implemented by us directly from the original algorithmic descriptions given in the previous sections.

All experiments were run on Ubuntu~22.04.4 using an
Intel\textsuperscript{\textregistered} Core\texttrademark{} i9-14900K processor
(24 cores, 32 threads) with 32\,GB RAM. Reported runtimes are wall-clock times for
sketch construction only.\\

\textbf{Baselines.}
The sketching methods compared in this section are listed below, along with brief descriptions.
\begin{itemize}

\item \textbf{Real Gaussian/Rademacher Sketch~\citep{kar2012random}:}
Dense JL-type polynomial sketches using i.i.d.\ Gaussian or Rademacher
projection matrices.

\item \textbf{Real \texttt{TensorSketch}~\citep{pham2013fast,pham2025tensorsketchfastscalable}:}
Standard CountSketch-based \texttt{TensorSketch} for polynomial kernels with
real random sign functions.

\item \textbf{CtR Gaussian/Rademacher Sketch~\citep{pmlr-v206-wacker23a}:}
\emph{complex-to-real} JL-type sketches using complex Gaussian or complex Rademacher
projections followed by CtR conversion, as proposed by~\citep{pmlr-v206-wacker23a}.

\item \textbf{CtR \texttt{TensorSketch} (Our Proposal):}
Our CtR \texttt{TensorSketch} obtained by replacing real random signs with
random function drawn uniformly from the four fourth roots of unity and applying
\emph{complex-to-real} conversion within the corresponding constructions.

\end{itemize}

\noindent \textbf{Datasets.}
We use both synthetic and real-world datasets.\\
\emph{Synthetic data.}
For the synthetic experiments, we generate $n=3000$ random vectors in
$\mathbb{R}^{d}$ with $d=2$.
Each coordinate is drawn independently from a standard Gaussian distribution,
and all vectors are $\ell_2$--normalized.
We consider polynomial kernels of degrees
$p \in \{10,15,20\}$ and vary the sketch dimension as
$D \in \{100,300,500\}$.

\emph{Real--world data.}
We evaluate all methods on two real-world datasets: 
MAGIC Gamma Telescope~\citep{magic_gamma_telescope_159} ($d=10$ real-valued features) 
and COD-RNA~\citep{articlecodrna} ($d=8$ numerical attributes). 
All inputs are treated as real-valued vectors and $\ell_2$-normalized prior to sketching. 
For each dataset, we subsample up to $n=3000$ points (or use the full dataset if smaller). 
Each configuration (dataset, degree $p$, sketch dimension $D$) is averaged over $20$ independent random trials.

\textbf{Comparison Metrics.}
We evaluate sketch quality and efficiency using two metrics.
\begin{itemize}
    \item \emph{KL divergence.}
    To assess how well a sketch preserves the structure of the degree-$p$
    polynomial kernel matrix, we compute the KL divergence between the exact
    kernel $K$ and its sketch-based approximation $\widehat{K}$.
    Since polynomial kernels are nonnegative, we normalize both matrices into
    discrete distributions
    \[
    P_{ij} = \frac{K_{ij}}{\sum_{a,b} K_{ab}},
    \qquad
    Q_{ij} = \frac{\widehat{K}_{ij}}{\sum_{a,b} \widehat{K}_{ab}},
    \]
    and report
    \[
    \mathrm{KL}(K \,\|\, \widehat{K})
    = \sum_{i,j} P_{ij}
    \log\!\left(\frac{P_{ij}}{Q_{ij}+\varepsilon}\right),
    \]
    where $\varepsilon>0$ ensures numerical stability. Lower values indicate
    better kernel preservation.

    \item \emph{Wall--clock time.}
    We measure efficiency by the wall--clock time required to construct sketch
    features for a given method and sketch dimension $D$, averaged over multiple
    independent trials.
\end{itemize}

\medskip 

\noindent \textbf{Results.}
Figures~\ref{fig:codrna_kl_sketch} and~\ref{fig:codrna_time_sketch} compare approximation error (KL divergence) and sketch construction time on the COD--RNA dataset across polynomial degrees and sketch dimensions. We evaluate the proposed CtR \texttt{TensorSketch} against the standard \texttt{TensorSketch}, along with real and CtR Gaussian and Rademacher JL-type sketches.

CtR \texttt{TensorSketch} consistently achieves the lowest KL divergence across all degrees and sketch dimensions, indicating more accurate kernel approximation than real \texttt{TensorSketch} and dense JL-type methods. This behavior matches our theoretical guarantee of  improved variance  established earlier.

In terms of runtime, CtR \texttt{TensorSketch} retains the input-sparsity cost $O(p(\nnz{\vec{x}} + D \log D))$, comparable to standard \texttt{TensorSketch}. In contrast, JL-type methods incur higher dense projection cost $O(p D d)$.

Additional experimental results are provided in Section~\ref{Further_exp} of the Appendix.

\section{Conclusion}\label{sec:conclusion}

We introduced a \emph{complex-to-real (CtR)} variant of \texttt{TensorSketch} that provides a compression algorithm for high-dimensional polynomial-kernel and tensor datasets.
 To the best of our knowledge, CtR constructions were only known for dense JL-type \texttt{TensorSketch}~\citep{pham2013fast} due to~\cite{pmlr-v206-wacker23a}.
Our results demonstrate that Complex-to-Real (CtR) variants achieve improved variance bounds compared to their real counterparts~\citep{pham2013fast}, matching those obtained in~\cite{pmlr-v206-wacker23a}. Moreover, these variants preserve the input-sparsity running time of the original method, making it a preferred choice over ~\cite{pmlr-v206-wacker23a}, for high-dimensional sparse data.

Two directions remain open for further investigation. First, it is unclear whether employing higher-order roots of unity can yield stronger higher-moment concentration, leading to tighter
 $(\epsilon, \delta)$ -approximation guarantees for the sketch. Second, it remains to be explored whether the advantages of Complex-to-Real (CtR) constructions can be extended to other kernel families and randomized feature maps beyond polynomial kernels. Together, these questions point to a broader potential for complex-valued sketch design.

\subsubsection*{Acknowledgements}

This research was partially supported by the ANRF under Project No. ANRF/ECRG/2024/001063/ENS. We gratefully acknowledge ANRF for this support.

\bibliography{uai2026-template}

\newpage

\onecolumn

\wideappendixtrue 

\title{Improving TensorSketch  Using Complex Random Variables\\(Supplementary Material)}
\maketitle


\appendix
\section{Proofs of Theorems in Section~\ref{sec:ctr_sketches}}

\fourththroot*
\begin{proof}
Since $s(i)$ is uniformly distributed over $\{1, \omega, \omega^{2}, \omega^{3}\} = \{1, i, -1, -i\}$, we have
\begin{align*}
    \Expt{s(i)}
    = \frac{1 + i - 1 - i}{4}
    = 0.
\end{align*}
All four values have unit magnitude, and therefore
\begin{align*}
    \Exptnb{\modulusSq{s(i)}}
    = \frac{\modulusSq{1} + \modulusSq{i} + \modulusSq{-1} + \modulusSq{-i}}{4}
    = 1.
\end{align*}
For the second complex moment,
\begin{align*}
    \Expt{s(i)^{2}}
    = \frac{1^{2} + i^{2} + (-1)^{2} + (-i)^{2}}{4}
    = \frac{1 - 1 + 1 - 1}{4}
    = 0.
\end{align*}
Finally, for any $i \neq j$, independence of $s(i)$ and $s(j)$ implies
\begin{align*}
    \Exptnb{s(i)\,\conj{s(j)}}
    = \Exptnb{s(i)}\,\Expt{\conj{s(j)}}
    = 0.
\end{align*}
\end{proof}

\subsection{Unbiasedness and Variance Analysis of CtR \texttt{TensorSketch}}\label{CtR_tensorsketch_appendix}
We establish the theoretical guarantees of CtR \texttt{TensorSketch}. 
Specifically, we show that the estimator is unbiased for the inner product estimation and subsequently provide a variance analysis.

\ctrTensorSketch*
\begin{proof}
According to Definition~\ref{CtR_Framework}, we have 
$\kCtRxy := \Re{\kCxy}$, and therefore we first analyze $\kCxy$. 
Recall that $\kCxy := \phiCx \overline{\phiCy}^{\top}$, where 
$\phiCx := \matrix{C}\vec{x}^{\otimes p}$ and 
$\phiCy := \matrix{C}\vec{y}^{\otimes p}$ according to 
Definition~\ref{def:ctr-tensorsketch}. 
In particular, $\matrix{C}\vec{x}^{\otimes p}, 
\matrix{C}\vec{y}^{\otimes p} \in \mathbb{C}^{D/2}$ denote the complex 
\texttt{CountSketch} of the vectors 
$\vec{x}^{\otimes p}, \vec{y}^{\otimes p} \in \mathbb{R}^{d^{p}}$, 
constructed using the derived functions 
$H : [d]^p \mapsto [D/2]$ and 
$S : [d]^p \to \{1,\omega,\omega^{2},\omega^{3}\}$ defined as
\begin{align*}
H(i_1,\ldots,i_p) &= \left( \sum_{j=1}^{p} h_j(i_j) \right) \bmod D/2, \\
S(i_1,\ldots,i_p) &= \prod_{j=1}^{p} s_j(i_j).
\end{align*}

\noindent In the proof, we denote $X := \vec{x}^{\otimes p}$ and $Y := \vec{y}^{\otimes p}$. 
Let $u,v \in [d]^p$ be the multi-indices corresponding to the entries of 
$X,Y \in \mathbb{R}^{d^p}$. For a multi-index $u=(j_1,\dots,j_p)$, the entry 
of $X$ is defined as $X_u = \prod_{k=1}^{p} x_{j_k}$, where the associated 
linearized index is $u = 1 + \sum_{k=1}^{p}(j_k - 1)\prod_{k'=1}^{p-k} d$. 
The entries of $Y$ are defined analogously. 

We begin by analyzing $\kCxy$ as defined below,
\begin{align*}
\kCxy &:= \phiCX \overline{\phiCY}^{\top} = \ip{\matrix{C} X}{\conj{\matrix{C}Y}}, \\
  &= \sum_{u,v \in [d]^p}
      X_{u}\, Y_{v}\, 
      S(u)\, \overline{S(v)}\ \indic{H(u) = H(v)}, \\
  &= \ \langle X, Y \rangle
   \!+\!
     \sum_{u \ne v}
      X_{u}\, Y_{v}\,
      S(u)\, \overline{S(v)}\, \indic{H(u) = H(v)}.
\end{align*}
Using Lemma~\ref{lem:sign_fun_4th_root}, we have 
$\Expt{S(u)\,\conj{S(v)}} = 0$ for all $u \neq v$. Hence,
\begin{align*}
\Expt{\kCxy}
= \ip{X}{Y} = \ip{\mathbf{x}^{\otimes p}}{\mathbf{y}^{\otimes p}}
= \ip{\vec{x}}{\vec{y}}^p .
\end{align*}
According to Definition~\ref{CtR_Framework}, we know $\kCtRxy := \Re{\kCxy}$, we have
\begin{align}
\Expt{\kCtRxy} = \Expt{\Re{\kCxy}} = \langle \vec{x}, \vec{y} \rangle^{p}. \label{appendix:eq:expectation-tensorsketch}
\end{align}
\noindent Using Equation~\ref{eq:ctr_variance_structure}, the variance of $\kCtRxy$ is
\begin{align} &\Var{\kCtRxy} = \frac{1}{2}\operatorname{Re}\!\Bigg\{ \Expt{ \modulusSq{\kCxy} } +  \Expt{ \left( \kCxy \right)^2 } - 2 \left( \Expt{ \kCxy } \right)^2 \Bigg\}. \label{appendix:eq:final variance ctr tensorsketch main} \end{align}

For the variance analysis, we need to compute 
$\Expt{\modulusSq{\kCxy}}$ and $\Expt{\left(\kCxy\right)^{2}}$.

\noindent Therefore by expanding 
$| \kCxy|^{2}$ we get:
\begin{align}
&| \kCxy|^{2}
:= |\langle \mathrm{C}\vec{x}^{\otimes p}, \overline{\mathrm{C}\vec{y}^{\otimes p}} \rangle| =\langle \mathrm{C}\vec{x}^{\otimes p}, \overline{\mathrm{C}\vec{y}^{\otimes p}} \rangle\langle \overline{\mathrm{C}\vec{x}^{\otimes p}}, \mathrm{C}\vec{y}^{\otimes p} \rangle, \\
=& \left(
    \langle X, Y \rangle
    + \sum_{u \ne v}
        X_{u} Y_{v} 
        S(u) \overline{S(v)}
        \indic{H(u) = H(v)}
  \right)\!\! \left(
    \langle X, Y \rangle
    + \sum_{u \ne v}
        X_{u} Y_{v} 
        \overline{S(u)} S(v)
        \indic{H(u) = H(v)}
  \right),\\
=& \langle X, Y \rangle^{2}
+  \langle X, Y \rangle 
    \left(\sum_{u \ne v}
        X_{u} Y_{v} 
        S(u) \overline{S(v)} 
        \indic{H(u) = H(v)} + \sum_{u \ne v}
        X_{u} Y_{v} 
        \overline{S(u)} S(v) 
        \indic{H(u) = H(v)}\right) \cdots \notag\\
&\qquad+
\left|\left(
\sum_{u \ne v}
    X_{u} Y_{v} 
    S(u) \overline{S(v)} 
    \indic{H(u) = H(v)}
\right)\right|^{2}.
\end{align}

Therefore, by applying expectation
\begin{align}
&\mathbb{E}\!\left[
\left|\left(
\sum_{u \ne v}
    X_{u} Y_{v} 
    S(u) \overline{S(v)} 
    \indic{H(u) = H(v)}
\right)\right|^{2}
\right] \notag
\\
&=
\mathbb{E}\!\left[
\sum_{\substack{u_1 \ne v_1 \\ u_2 \ne v_2}}
    X_{u_1} Y_{v_1}
    X_{u_2} Y_{v_2}
    S(u_1) \overline{S(v_1)}
    \overline{S(u_2)} S(v_2)
    \indic{H(u_{1}) = H(v_{1})}
    \indic{H(u_{2}) = H(v_{2})}
\right], \\
&=
\sum_{\substack{u_1 \ne v_1 \\ u_2 \ne v_2}}
\mathbb{E}\!\left[
    X_{u_1} Y_{v_1}
    X_{u_2} Y_{v_2}
    S(u_1) \overline{S(v_1)}
    \overline{S(u_2)} S(v_2)
\right]
\cdot
\mathbb{E}[\indic{H(u_{1}) = H(v_{1})}
    \indic{H(u_{2}) = H(v_{2})}], \\
&\le
\frac{2}{D}
\sum_{\substack{u_1 \ne v_1 \\ u_2 \ne v_2}}
\mathbb{E}\!\left[
    X_{u_1} Y_{v_1}
    X_{u_2} Y_{v_2}
    S(u_1) \overline{S(v_1)}
    \overline{S(u_2)} S(v_2)
\right], \\
&\le
\frac{2}{D}
\sum_{u_1 \ne v_1}
\mathbb{E}\!\left[
    |X_{u_1}|\, |Y_{v_1}|\,
    |X_{u_2}|\, |Y_{v_2}|\,
    S(u_1) \overline{S(v_1)}
    \overline{S(u_2)} S(v_2)
\right], \\
&=
\frac{2}{D}
\mathbb{E}\!\left[
\left|\left(
\sum_{u \ne v \in [d]^p}
    |X_{u}|\, |Y_{v}|\,
    S(u) \overline{S(v)}
\right)\right|^{2}
\right].\label{appendix:eq:second moment norm ctr}
\end{align}

We now bound the above expression using the second-moment bound from Lemma~\ref{lem:complex_to_real AMS}, given in Equation~\eqref{appendix:eq:modulo second moment x^p}. For completeness, we first restate the bound,
\begin{align}
\mathbb{E}\left[\left|\left(
\sum_{u , v \in [d]^p}
    |X_{u}|\, |Y_{v}|\,
    S(u) \overline{S(v)}
\right)\right|^{2}\right]
= \left(\langle \vec{x}, \vec{y} \rangle^{2} + \|\vec{x}\|_2^{2}\, \|\vec{y}\|_2^{2}
- \sum_{i=1}^{d} x_i^{2} y_i^{2}\right)^{p}. \label{appendix:eq:exact bound tensor sketch z ctr}
\end{align}
Now, we expand the term $\left|\left(
\sum_{u , v \in [d]^p}
    |X_{u}|\, |Y_{v}|\,
    S(u) \overline{S(v)}
\right)\right|^{2}$
from the above  Equation~\eqref{appendix:eq:exact bound tensor sketch z ctr} as follows,
\begin{align}
    &\left|\left(
\sum_{u , v \in [d]^p}
    |X_{u}|\, |Y_{v}|\,
    S(u) \overline{S(v)}
\right)\right|^{2} = \left| \sum_{u  \in [d]^p}
    |X_{u}|\, |Y_{v}|  + \sum_{\substack{u, v \in [d]^p \\ u \neq v}}
    |X_{u}|\, |Y_{v}|\,
    S(u) \overline{S(v)}  \right|^{2},\\
&= \left(
\sum_{u \in [d]^p}
    |X_{u}|\, |Y_{u}|\, 
\;+\;
\sum_{\substack{u, v \in [d]^p \\ u \neq v}}
    |X_{u}|\, |Y_{v}|\,
    S(u) \overline{S(v)}  
\right) \overline{\left(
\sum_{w \in [d]^p}
    |X_{w}|\, |Y_{w}|\,
\;+\;
\sum_{\substack{w, z \in [d]^p \\ w \neq z}}
    |X_{w}|\, |Y_{z}|\,
    S(w) \overline{S(z)}  
\right)}, \\
&= \left(
\sum_{u \in [d]^p}
    |X_{u}|\, |Y_{u}|\, 
\;+\;
\sum_{\substack{u, v \in [d]^p \\ u \neq v}}
    |X_{u}|\, |Y_{v}|\,
    S(u) \overline{S(v)}  
\right)  \left(
\sum_{w \in [d]^p}
    |X_{w}|\, |Y_{w}|\,
\;+\;
\sum_{\substack{w, z \in [d]^p \\ w \neq z}}
    |X_{w}|\, |Y_{z}|\,
    \overline{S(w)} S(z)  
\right),\\
&=
\sum_{u, w \in [d]^p}
    |X_{u}|\, |Y_{u}|\,
    |X_{w}|\, |Y_{w}|\,  +
\sum_{u \in [d]^p}
\sum_{\substack{w, z \in [d]^p \\ w \neq z}}
    |X_{u}|\, |Y_{u}|\,
    |X_{w}|\, |Y_{z}|\,
    \overline{S(w)} S(z) + \cdots \notag
\\
&\cdots 
+
\sum_{\substack{u, v \in [d]^p \\ u \neq v}}
\sum_{w \in [d]^p}
    |X_{u}|\, |Y_{v}|\,
    |X_{w}|\, |Y_{w}|\,
    S(u) \overline{S(v)}\, 
+\sum_{\substack{u, v \in [d]^p \\ u \neq v}}
\sum_{\substack{w, z \in [d]^p \\ w \neq z}}
    |X_{u}|\, |Y_{v}|\,
    |X_{w}|\, |Y_{z}|\,
    S(u) \overline{S(v)}\,
    \overline{S(w)} S(z). \label{appendix:eq:second moment tensor sketch}
\end{align}
Using Lemma~\ref{lem:sign_fun_4th_root}, we have 
$\Expt{S(u)\,\conj{S(v)}} = 0$ for all $u \neq v$. Hence,
\begin{align}
    \sum_{\substack{u, v \in [d]^p \\ u \neq v}}
    \sum_{w \in [d]^p}
        |X_{u}|\, |Y_{v}|\,
        |X_{w}|\, |Y_{w}|\,
        \mathbb{E}\!\left[S(u) \overline{S(v)}\right]
        = 0,
\end{align}

\item Similarly, for all $w \neq z$,
\begin{align}
    \sum_{u \in [d]^p}
    \sum_{\substack{w, z \in [d]^p \\ w \neq z}}
        |X_{u}|\, |Y_{u}|\,
        |X_{w}|\, |Y_{z}|\,
        \mathbb{E}\!\left[\overline{S(w)} S(z)\right]
        = 0,
\end{align}

Substitute this in Equation~\eqref{appendix:eq:exact bound tensor sketch z ctr}, we get
\begin{align}
     \mathbb{E} \left[ \left|
\sum_{u , v \in [d]^p}
    |X_{u}|\, |Y_{v}|\,
    S(u) \overline{S(v)}
\right|^{2} \right] &=  \sum_{u, w \in [d]^p}
    |X_{u}|\, |Y_{u}|\,
    |X_{w}|\, |Y_{w}|\ + \cdots  \notag\\& \cdots +\mathbb{E}\left[\sum_{\substack{u, v \in [d]^p \\ u \neq v}}
\sum_{\substack{w, z \in [d]^p \\ w \neq z}}
    |X_{u}|\, |Y_{v}|\,
    |X_{w}|\, |Y_{z}|\,
    S(u) \overline{S(v)}\,
    \overline{S(w)} S(z)\right],\\
    &=  \left\langle X,Y\right\rangle^{2} +\mathbb{E}  \left|\left(
\sum_{u \neq v }
    |X_{u}|\, |Y_{v}|\,
    S(u) \overline{S(v)}
\right)\right|^{2}.
\end{align}
We conclude that,
 \begin{align}
\mathbb{E} \left[ \left|
\sum_{u \neq v }
    |X_{u}|\, |Y_{v}|\,
    S(u) \overline{S(v)}
\right|^{2} \right]      &=  \mathbb{E}  \left[\left|
\sum_{u , v \in [d]^p}
    |X_{u}|\, |Y_{v}|\,
    S(u) \overline{S(v)}
\right|^{2}\right]  - \left\langle \mathbf{x,y}\right\rangle^{2p}.
\end{align}
\noindent Putting value of $\mathbb{E}\left[\left|\left(
\sum_{u , v \in [d]^p}
    |X_{u}|\, |Y_{v}|\,
    S(u) \overline{S(v)}
\right)\right|^{2}\right]$, we get  
\begin{align}
\mathbb{E} \left[ \left|
\sum_{u \neq v }
    |X_{u}|\, |Y_{v}|\,
    S(u) \overline{S(v)}
\right|^{2}\right]      &=  \left(\langle \vec{x}, \vec{y} \rangle^{2} + \|\vec{x}\|_2^{2}\, \|\vec{y}\|_2^{2}
- \sum_{i=1}^{d} x_i^{2} y_i^{2}\right)^{p}  - \left\langle \mathbf{x,y}\right\rangle^{2p}. \label{appendix:eq:exact bound for tensor sketch non diagonal terms}
\end{align}
Put this value in Equation~\ref{appendix:eq:second moment norm ctr},
\begin{align}
    \mathbb{E}\!\left[
\left|
\sum_{u \ne v}
    X_{u} Y_{v} 
    S(u) \overline{S(v)} 
    \indic{H(u) = H(v)}
\right|^{2}
\right] \le \frac{2}{D}\left(\left(\langle \vec{x}, \vec{y} \rangle^{2} + \|\vec{x}\|_2^{2}\, \|\vec{y}\|_2^{2}
- \sum_{i=1}^{d} x_i^{2} y_i^{2}\right)^{p}  - \left\langle \mathbf{x,y}\right\rangle^{2p}\right).
\end{align}
Now we have the second moment as follows, 
\begin{align}
  \mathbb{E}\left[  | \kCxy|^{2}\right] &= \left\langle X,Y \right\rangle^{2} +    \mathbb{E}\!\left[
\left|\left(
\sum_{u \ne v}
    X_{u} Y_{v} 
    S(u) \overline{S(v)} 
    \indic{H(u) = H(v)}
\right)\right|^{2}
\right],\\
&\le \left\langle \vec{x},\vec{y} \right\rangle^{2p} + \frac{2}{D}\left(\left(\langle \vec{x}, \vec{y} \rangle^{2} + \|\vec{x}\|_2^{2}\, \|\vec{y}\|_2^{2}
- \sum_{i=1}^{d} x_i^{2} y_i^{2}\right)^{p}  - \left\langle \mathbf{x,y}\right\rangle^{2p}\right).\label{appendix:second moment norm ctr}
\end{align}

\noindent Next, we analyze $\mathbb{E}\!\left[(\kCxy)^{2}\right]$. Expanding $(\kCxy)^{2}$ gives:
\begin{align}
&( \kCxy)^{2}
:= \left(\langle \mathrm{C}\vec{x}^{\otimes p}, \overline{\mathrm{C}\vec{y}^{\otimes p}} \rangle\right)^{2} = \langle \mathrm{C}\vec{x}^{\otimes p}, \overline{\mathrm{C}\vec{y}^{\otimes p}} \rangle \langle \mathrm{C}\vec{x}^{\otimes p}, \overline{\mathrm{C}\vec{y}^{\otimes p}} \rangle,\\
&= \left(
    \langle X, Y \rangle
    + \sum_{u \ne v}
        X_{u} Y_{v} 
        S(u) \overline{S(v)}
        \indic{H(u) = H(v)}
  \right)\left(
    \langle X, Y \rangle
    + \sum_{u \ne v}
        X_{u} Y_{v} 
        S(u) \overline{S(v)}
        \indic{H(u) = H(v)}
  \right),\\
&= \langle X, Y \rangle^{2}
+  2\langle X, Y \rangle 
    \left(\sum_{u \ne v}
        X_{u} Y_{v} 
        S(u) \overline{S(v)} 
        \indic{H(u) = H(v)}\right) +\left(
\sum_{u \ne v}
    X_{u} Y_{v} 
    S(u) \overline{S(v)} 
    \indic{H(u) = H(v)}
\right)^{2},
\end{align}

Therefore, by applying expectation
\begin{align}
&\mathbb{E}\!\left[
\left(
\sum_{u \ne v}
    X_{u} Y_{v} 
    S(u) \overline{S(v)} 
    \indic{H(u) = H(v)}
\right)^{2}
\right]\\
&=
\mathbb{E}\!\left[
\sum_{\substack{u_1 \ne v_1 \\ u_2 \ne v_2}}
    X_{u_1} Y_{v_1}
    X_{u_2} Y_{v_2}
    S(u_1) \overline{S(v_1)}
    S(u_2) \overline{S(v_2)}
    \indic{H(u_{1}) = H(v_{1})}
    \indic{H(u_{2}) = H(v_{2})}
\right], \\
&=
\sum_{\substack{u_1 \ne v_1 \\ u_2 \ne v_2}}
\mathbb{E}\!\left[
    X_{u_1} Y_{v_1}
    X_{u_2} Y_{v_2}
    S(u_1) \overline{S(v_1)}
    S(u_2) \overline{S(v_2)}
\right]
\cdot
\mathbb{E}[\indic{H(u_{1}) = H(v_{1})}
    \indic{H(u_{2}) = H(v_{2})}], \\
&\leq
\frac{2}{D}
\mathbb{E}\!\left[
\left(
\sum_{u \ne v \in [d]^p}
    X_{u}\, Y_{v}\,
    S(u) \overline{S(v)}
\right)^{2}
\right].
\end{align}

\noindent By similar analysis and using second-moment bound of Lemma~\ref{lem:complex_to_real AMS} given in Equation~\eqref{appendix:eq:second moment x^p}, we get
\begin{align}
    \mathbb{E}\!\left[
\left(
\sum_{u \ne v}
    X_{u} Y_{v} 
    S(u) \overline{S(v)} 
    \indic{H(u) = H(v)}
\right)^{2}
\right] \le \frac{2}{D}\left(\left(2\langle \vec{x}, \vec{y} \rangle^{2} - \sum_{i=1}^{d} x_i^{2} y_i^{2}\right)^{p}  - \left\langle \mathbf{x,y}\right\rangle^{2p}\right).
\end{align}
Therefore, we have,
\begin{align}
  \mathbb{E}\left[  \left( \kCxy\right)^{2}\right] &= \left\langle X,Y \right\rangle^{2} +    \mathbb{E}\!\left[
\left(
\sum_{u \ne v}
    X_{u} Y_{v} 
    S(u) \overline{S(v)} 
    \indic{H(u) = H(v)}
\right)^{2}
\right],\\
&\le \left\langle \vec{x},\vec{y} \right\rangle^{2p} + \frac{2}{D}\left(\left(2\langle \vec{x}, \vec{y} \rangle^{2} - \sum_{i=1}^{d} x_i^{2} y_i^{2}\right)^{p}  - \left\langle \mathbf{x,y}\right\rangle^{2p}\right).\label{appendix:eq:second moment ctr}
\end{align}

\noindent Therefore for final variance we put Equation~\eqref{appendix:second moment norm ctr} and Equation~\eqref{appendix:eq:second moment ctr} in Equation~\eqref{appendix:eq:final variance ctr tensorsketch main},we get
\begin{align}
    \Var{\kCtRxy} 
    &\le \frac{1}{2}\Bigg[ \langle \vec{x},\vec{y} \rangle^{2p} + 
        \frac{2}{D}\Big( \big(\langle \vec{x}, \vec{y} \rangle^{2} + \|\vec{x}\|_2^{2}\, \|\vec{y}\|_2^{2}
        - \sum_{i=1}^{d} x_i^{2} y_i^{2}\big)^{p}  - \langle \vec{x},\vec{y} \rangle^{2p} \Big) \nonumber \\
    &\qquad\quad
        + \langle \vec{x},\vec{y} \rangle^{2p} +\frac{2}{D}\Big( \big(2\langle \vec{x}, \vec{y} \rangle^{2} - \sum_{i=1}^{d} x_i^{2} y_i^{2}\big)^{p}  
        - \langle \vec{x},\vec{y} \rangle^{2p} \Big) -2\langle \vec{x},\vec{y} \rangle^{2p}
    \Bigg],\\
    \intertext{we can upper bound the above equation by using inequality $\langle\mathbf{x,y}\rangle^{2} \leq \|\mathbf{x}\|^{2}_{2}\|\mathbf{y}\|^{2}_{2}$, then}\notag\\
    &\leq \frac{1}{2}\left[\frac{2}{D}\left(\left( 2\|\vec{x}\|_2^{2}\, \|\vec{y}\|_2^{2} \right)^{p}  - \|\vec{x}\|_2^{2p}\, \|\vec{y}\|_2^{2p}\right) + \frac{2}{D}\left(\left( 2\|\vec{x}\|_2^{2}\, \|\vec{y}\|_2^{2} \right)^{p}  - \|\vec{x}\|_2^{2p}\, \|\vec{y}\|_2^{2p}\right)\right],\\
    &\leq \frac{1}{2}\left[\frac{2^{p+1} - 2}{D}\|\vec{x}\|_{2}^{2p}\,\|\vec{y}\|_{2}^{2p} 
    + \frac{2^{p+1}- 2}{D}\|\vec{x}\|_{2}^{2p}\,\|\vec{y}\|_{2}^{2p}\right],\\
    &\leq \frac{2^{p+1} - 2}{D}\|\vec{x}\|_{2}^{2p}\,\|\vec{y}\|_{2}^{2p}.
\end{align}
\end{proof}
\noindent We begin by stating a lemma that serves as a key step to bound the second moments of the estimator in the proof of Theorem~\ref{thm:complex_countsketch_degree_p} .
\AMSSketch*
\begin{proof} First, we consider the expectation. For each $j$, we note that
\begin{align}
\mathbb{E}\!\left[ Z_{s_j}(\vec{x})\, \overline{Z_{s_j}(\vec{y})} \right]
&= 
\mathbb{E}\!\left[
\left( \sum_{i=1}^{d} x_i\, s_j(i) \right)
\left( \sum_{k=1}^{d} y_k\, \overline{s_j(k)} \right)
\right], \\
&= \sum_{i=1}^{d}\sum_{k=1}^{d} x_i y_k\, \mathbb{E}[s_j(i)\overline{s_j(k)}], \\
&= \sum_{i=1}^{d} x_i y_i\, \mathbb{E}[|s_j(i)|^2] +  \sum_{i\neq k} x_i y_k\, \mathbb{E}[s_j(i)\overline{s_j(k)}], \\ 
&= \langle \vec{x}, \vec{y} \rangle,
\end{align}
where, $\mathbb{E}[s_j(i)\overline{s_j}(k)] = 0 , \forall i\neq k$ and $\mathbb{E}[|s_j(i)|^2] = 1,  \forall\  i\in [d]$ as given in Lemma~\ref{lem:sign_fun_4th_root}.

\noindent Since the functions $s_j$ are independent across different $j$, we have
\begin{align}
\mathbb{E}[Z] 
= \prod_{j=1}^{p} \mathbb{E}[Z_{s_j}(\vec{x}) \overline{Z_{s_j}(\vec{y})}]
= \langle \vec{x}, \vec{y} \rangle^{p}.
\end{align}

\noindent Next, to bound the variance,
\begin{align}
\Var{Z} = \frac{1}{2}\operatorname{Re}\!\big\{
\mathbb{E}[|Z|^2] + \mathbb{E}[(Z)^2] - 2|\mathbb{E}[Z]|^2
\big\}.
\end{align}

\noindent Because the random hash functions $s_j$ are independent across different $j$, we may write
\begin{align}
\mathbb{E}[|Z|^{2}]
= \prod_{j=1}^{p} \mathbb{E}\!\left[ |\left( Z_{s_j}(\vec{x})\, \overline{Z_{s_j}(\vec{y})} \right)|^{2} \right] \text{ and } \mathbb{E}[(Z)^{2}]
= \prod_{j=1}^{p} \mathbb{E}\!\left[ \left( Z_{s_j}(\vec{x})\, \overline{Z_{s_j}(\vec{y})} \right)^{2} \right] \label{ams norm and square}
\end{align}

\noindent For each $j$, expanding the square gives
\begin{align}
\mathbb{E}\!\left[ |\left( Z_{s_j}(\vec{x})\, \overline{Z_{s_j}(\vec{y})} \right)|^{2} \right]
&=
\mathbb{E}\!\left[
\left( \sum_{i=1}^{d} x_i\, s_j(i) \right)
\left( \sum_{k=1}^{d} y_k\, \overline{s_j(k)} \right)\left( \sum_{i=1}^{d} x_i\, \overline{s_j(i)} \right)
\left( \sum_{k=1}^{d} y_k\, s_j(k) \right)
\right] ,\\
&=
\sum_{i=1}^{d} \sum_{i'=1}^{d} \sum_{k=1}^{d} \sum_{k'=1}^{d}
x_i\, x_{i'}\, y_k\, y_{k'}\;
\mathbb{E}\!\left[ s_j(i)\overline{s_j(k)}\overline{s_j(i')}s_j(k') \right].
\end{align}

\noindent Observing that 
$\mathbb{E}[s_j(i)\overline{s_j(k)}\overline{s_j(i')}s_j(k')]$ 
is nonzero only when the indices form pairs (including the possibility that all four are identical), we have
\begin{align}
\mathbb{E}[s_j(i)\overline{s_j(k)}\overline{s_j(i')}s_j(k')]
=
\begin{cases}
1, & \text{if } i = k = i' = k', \\[4pt]
1, & \text{if } i = k \ne i' = k', \\[4pt]
1, & \text{if } i = i' \ne k = k', \\[4pt]
0, & \text{otherwise}.
\end{cases}
\end{align}
\noindent The contribution from terms with $i = k = i' = k'$ is 
\begin{align}
\sum_{i=1}^{d} x_i^{2} y_i^{2}.
\end{align}
\noindent Terms with $i = k \ne i' = k'$ contribute
\begin{align}
\sum_{i \ne i'} x_i y_i\, x_{i'} y_{i'}
= \left( \sum_{i=1}^{d} x_i y_i \right)^{2}
 - \sum_{i=1}^{d} x_i^{2} y_i^{2}
= \langle \vec{x}, \vec{y} \rangle^{2} - \sum_{i=1}^{d} x_i^{2} y_i^{2}.
\end{align}
Finally, for $i = i' \ne k = k'$ we obtain
\begin{align}
\sum_{i \ne k} x_i^{2} y_k^{2}
= \|\vec{x}\|_2^{2}\, \|\vec{y}\|_2^{2}
 - \sum_{i=1}^{d} x_i^{2} y_i^{2}.
\end{align}
Thus, summing these contributions, we have
\begin{align}
\mathbb{E}\!\left[| \left( Z_{s_j}(\vec{x})\, Z_{s_j}(\vec{y}) \right) |^{2} \right]
&=
\sum_{i=1}^{d} x_i^{2} y_i^{2}
+ \left( \langle \vec{x}, \vec{y} \rangle^{2} + \|\vec{x}\|_2^{2}\, \|\vec{y}\|_2^{2}
 - 2\sum_{i=1}^{d} x_i^{2} y_i^{2} \right),
\\[4pt]
&=
 \langle \vec{x}, \vec{y} \rangle^{2} + \|\vec{x}\|_2^{2}\, \|\vec{y}\|_2^{2}
- \sum_{i=1}^{d} x_i^{2} y_i^{2}.
\end{align}
Substituting this bound into Equation~\eqref{ams norm and square} yields
\begin{align}    
\mathbb{E}[|Z|^{2}]
= \left(\langle \vec{x}, \vec{y} \rangle^{2} + \|\vec{x}\|_2^{2}\, \|\vec{y}\|_2^{2}
- \sum_{i=1}^{d} x_i^{2} y_i^{2}\right)^{p}, \label{appendix:eq:modulo second moment x^p}
\end{align}

\noindent For each $j$, expanding the square gives
\begin{align}
\mathbb{E}\!\left[ \left( Z_{s_j}(\vec{x})\, \overline{Z_{s_j}(\vec{y})} \right)^{2} \right]
&=
\mathbb{E}\!\left[
\left( \sum_{i=1}^{d} x_i\, s_j(i) \right)
\left( \sum_{k=1}^{d} y_k\, \overline{s_j(k)} \right)\left( \sum_{i=1}^{d} x_i\, s_j(i) \right)
\left( \sum_{k=1}^{d} y_k\, \overline{s_j(k)} \right)
\right] ,\\
&=
\sum_{i=1}^{d} \sum_{i'=1}^{d} \sum_{k=1}^{d} \sum_{k'=1}^{d}
x_i\, x_{i'}\, y_k\, y_{k'}\;
\mathbb{E}\!\left[ s_j(i)\overline{s_j(k)}s_j(i')\overline{s_j(k')} \right].
\end{align}

\noindent Observing that 
$\mathbb{E}[s_j(i)\overline{s_j(k)}s_j(i')\overline{s_j(k')}]$ 
is nonzero only when the indices form pairs (including the possibility that all four are identical), we have
\begin{align}
\mathbb{E}[s_j(i)\overline{s_j(k)}s_j(i')\overline{s_j(k')}]
=
\begin{cases}
1, & \text{if } i = k = i' = k', \\[4pt]
1, & \text{if } i = k \ne i' = k', \\[4pt]
1, & \text{if } i = k' \ne i' = k, \\[4pt]
0, & \text{otherwise}.
\end{cases}
\end{align}
\noindent The contribution from terms with $i = k = i' = k'$ is 
\begin{align}
\sum_{i=1}^{d} x_i^{2} y_i^{2}.
\end{align}
\noindent Terms with $i = k \ne i' = k'$ contribute
\begin{align}
\sum_{i \ne i'} x_i y_i\, x_{i'} y_{i'}
= \left( \sum_{i=1}^{d} x_i y_i \right)^{2}
 - \sum_{i=1}^{d} x_i^{2} y_i^{2}
= \langle \vec{x}, \vec{y} \rangle^{2} - \sum_{i=1}^{d} x_i^{2} y_i^{2}.
\end{align}
Finally, for $i = k' \ne i' = k$ we obtain
\begin{align}
\sum_{i \ne i'} x_i y_i\, x_{i'} y_{i'}
= \left( \sum_{i=1}^{d} x_i y_i \right)^{2}
 - \sum_{i=1}^{d} x_i^{2} y_i^{2}
= \langle \vec{x}, \vec{y} \rangle^{2} - \sum_{i=1}^{d} x_i^{2} y_i^{2}.
\end{align}
Thus, summing these contributions, we have
\begin{align}
\mathbb{E}\!\left[ \left( Z_{s_j}(\vec{x})\, Z_{s_j}(\vec{y}) \right)^{2} \right]
&=
\sum_{i=1}^{d} x_i^{2} y_i^{2}
+ \left( 2\langle \vec{x}, \vec{y} \rangle^{2}- 2\sum_{i=1}^{d} x_i^{2} y_i^{2} \right),
\\[4pt]
&=
 2\langle \vec{x}, \vec{y} \rangle^{2}
- \sum_{i=1}^{d} x_i^{2} y_i^{2}.
\end{align}
Substituting this bound into Equation~\eqref{ams norm and square} yields
\begin{align}    
\mathbb{E}[(Z)^{2}]
= \left(2\langle \vec{x}, \vec{y} \rangle^{2}
- \sum_{i=1}^{d} x_i^{2} y_i^{2} \right)^p, \label{appendix:eq:second moment x^p}
\end{align}

\noindent which completes the proof since
\begin{align}
\Var{Z}
&= \frac{1}{2}\operatorname{Re}\!\big\{
\mathbb{E}[|Z|^2] + \mathbb{E}[(Z)^2] - 2|\mathbb{E}[Z]|^2
\big\},\\
&= \frac{1}{2}\left\{\left(\langle \vec{x}, \vec{y} \rangle^{2} + \|\vec{x}\|_2^{2}\, \|\vec{y}\|_2^{2}
- \sum_{i=1}^{d} x_i^{2} y_i^{2}\right)^{p} +  \left(2\langle \vec{x}, \vec{y} \rangle^{2}
- \sum_{i=1}^{d} x_i^{2} y_i^{2} \right)^p - 2\langle \vec{x}, \vec{y} \rangle^{2p}\right\}\label{appendix:eq:variance_CtR_AMS}.
\end{align}
Using the Cauchy--Schwarz inequality, 
$\langle \vec{x}, \vec{y} \rangle^{2} \le \|\vec{x}\|_2^{2}\, \|\vec{y}\|_2^{2}$,  
and noting that 
$\sum_{i=1}^{d} x_i^{2} y_i^{2} \ge 0$,  
it follows that
\begin{align}
\Var{Z} \leq  2^{p} \|\vec{x}\|_2^{2p}\|\vec{y}\|_2^{2p}.
\end{align}
\medskip

\end{proof}

\section{proof of \texttt{CountSketch} Estimator with Fourth Roots of Unity}

\begin{thm}
\label{thm:cs-unbiasedness}
Let $\Phi$ denote the CountSketch with the hash function taking values in the fourth roots of unity.  
For any vectors $\mathbf{x}, \mathbf{y} \in \mathbb{R}^d$, the corresponding sketch vector is defined by
\begin{align}
\Phi(\mathbf{x})_j 
= \sum_{i=1}^{d} \delta_{ji}\sigma(i) x_i\,, \qquad \Phi(\mathbf{y})_j 
= \sum_{i=1}^{d} \delta_{ji}\sigma(i) y_i,
\qquad \forall j \in [D/2].
\end{align}
where,
\begin{itemize}
    \item $h : [d] \to [D/2]$ assigning each coordinate independently and uniformly to one of $D/2$ buckets, and  
    \item $\sigma: [d] \to \{1, \omega, \omega^{2}, \omega^{3}\}$ are drawn independently and uniformly at random,
\item  $\delta_{j i} = 
\begin{cases}
1, & \text{if } h(i) = j, \\
0, & \text{otherwise.}
\end{cases}$
\end{itemize}
then the estimator $\hat{k}(\mathbf{x}, \mathbf{y})
:= \Re{\Phi(\mathbf{x})^\top \overline{\Phi(\mathbf{y})}},$
satisfies
\begin{align}
\mathbb{E}\!\left[\hat{k}(\mathbf{x}, \mathbf{y})\right]
&= \langle \mathbf{x}, \mathbf{y} \rangle,\\
\operatorname{Var}\!\big(\hat{k}(\mathbf{x}, \mathbf{y})\big)
&= \frac{1}{D} \!\left(
        \|\mathbf{x}\|_2^2 \|\mathbf{y}\|_2^2
        + \langle \mathbf{x}, \mathbf{y} \rangle^2
        - 2 \sum_{i=1}^d x_i^2 y_i^2
    \right).
\end{align}
\end{thm}

\begin{proof}
\noindent Since $\sigma(i)$ is uniformly distributed over $\{1, \omega, \omega^{2}, \omega^{3}\} = \{1, i, -1, -i\}$ and satisfies $\mathbb{E}[\sigma(i)] = 0, \,  \mathbb{E}\!\left[\,|\sigma(i)|^{2}\,\right] = 1, \,  \mathbb{E}\!\left[(\sigma(i))^{2}\right] =0$. Consider now the complex inner product of the sketch:
\begin{align}
\hat{k}_C(\mathbf{x}, \mathbf{y})
&:= \Phi(\mathbf{x})^\top \overline{\Phi(\mathbf{y})} = \sum_{j=1}^{D/2}
\Big(\sum_{i:\,h(i)=j} \sigma(i) x_i\Big)
\Big(\sum_{m:\,h(m)=j} \overline{\sigma(m)}\, y_m\Big)\\[4pt]
&= \sum_{i=1}^{d} |\sigma(i)|^2 x_i y_i
  + \sum_{i \neq m} \mathbf{1}\{h(i) = h(m)\}\, \sigma(i)\overline{\sigma(m)}\, x_i y_m.
\end{align}

Taking expectation with respect to $h$ and $\sigma$, we get
\begin{align}
\mathbb{E}\!\left[\hat{k}_C(\mathbf{x}, \mathbf{y})\right]
= \sum_{i=1}^{d} \mathbb{E}\!\left[|\sigma(i)|^2\right] x_i y_i
= \sum_{i=1}^{d} x_i y_i
= \langle \mathbf{x}, \mathbf{y} \rangle.
\label{eq:complex_inner_unbiased}
\end{align}

Since the right-hand side of \eqref{eq:complex_inner_unbiased} is real, the same unbiasedness holds for the estimator, which takes the real part of the complex sketch:
\begin{align}
\mathbb{E}\!\left[\hat{k}(\mathbf{x}, \mathbf{y})\right]
&= \mathbb{E}\left[\operatorname{Re}\left\{\hat{k}_C(\mathbf{x}, \mathbf{y})\right\}\right]
= \operatorname{Re}\left\{\mathbb{E}\left[\hat{k}_C(\mathbf{x}, \mathbf{y})\right]\right\}
= \langle \mathbf{x}, \mathbf{y} \rangle.
\end{align}

To find the exact variance of the estimator, we leverage the Complex-to-Real (CtR) framework stated in our paper. 
\noindent Now, we work out with \(\mathbb{E}\!\left[|\hat{k}_C(\mathbf{x}, \mathbf{y})|^2\right]\). \\
\begin{align}
    \mathbb{E}\!\left[|\hat{k}_C(\mathbf{x}, \mathbf{y})|^2\right]
&= \mathbb{E}\!\left[|\Phi(\mathbf{x})^\top \overline{\Phi(\mathbf{y})} |^2\right] =\mathbb{E}\!\left[\left|\sum_{j=1}^{D/2}\left(\sum_{i=1}^{d} \delta_{ji}\sigma(i)x_i\right)
\left(\sum_{m=1}^{d} \delta_{jm}\overline{\sigma(m)}\,y_m\right) \right|^2\right],\\ 
&= \sum_{j,j'}^{D/2}\sum_{i,m,i',m'}^d 
\mathbb{E}\!\big[\delta_{ji}\delta_{jm}\delta_{j'i}\delta_{j'm'}\big]\;
\mathbb{E}\!\big[\sigma(i)\overline{\sigma(m)}\,\overline{\sigma(i')} \sigma(m')\big]\;
x_i y_m x_{i'} y_{m'}.
\end{align}

Hence,
\[
\mathbb{E}\!\big[\sigma(i)\overline{\sigma(m)}\,\overline{\sigma(i')}\sigma(m')\big]\neq 0
\iff \{i,m,i',m'\}\ \text{appear in pairs},
\]
\begin{enumerate}
    \item With surviving index configurations, when $\mathbf{j = j'}$:
\begin{itemize}
    \item $ \mathbf{i =m =i' = m'}: \mathbb{E}[\delta_{ji}]\mathbb{E}[|\sigma(i)|^{4}]x_{i}^{2}y_{i}^{2}$,
    \item $ \mathbf{i =m \neq i' = m'}: \mathbb{E}[\delta_{ji}\delta_{ji'}]\mathbb{E}[|\sigma(i)|^{2}|\sigma(i')|^{2}]x_{i}y_{i}x_{i'}y_{i'}$,
    \item $ \mathbf{i =i' \neq m = m'}: \mathbb{E}[\delta_{ji}\delta_{jm}]\mathbb{E}[|\sigma(i)|^{2}|\sigma(m)|^{2}]x_{i}^{2}y_{m}^{2}$,
\end{itemize}

\item  With surviving index configurations, when $\mathbf{j \neq j'}$:
\begin{itemize}
    \item $ \mathbf{i =m \neq i' = m'}: \mathbb{E}[\delta_{ji}\delta_{j'i'}]\mathbb{E}[|\sigma(i)|^{2}|\sigma(i')|^{2}]x_{i}y_{i}x_{i'}y_{i'} $,
\end{itemize}
\end{enumerate}

\noindent Thus,
\begin{align}
\mathbb{E}\!\left[|\hat{k}_C(\mathbf{x}, \mathbf{y})|^2\right]
&= \sum_{i=1}^d x_i^2 y_i^2 + \frac{2}{D}\sum_{i\neq i'} (x_{i}y_{i}x_{i'}y_{i'} + x_{i}^{2}y_{i'}^{2}) + \frac{D-2}{D}\sum_{i\neq i'}x_{i}y_{i}x_{i'}y_{i'}, \label{norm_squared_part_1}\\
&= \langle \mathbf{x,y}\rangle^{2} + \frac{2}{D}\sum_{i\neq i'}  x_{i}^{2}y_{i'}^{2} . \label{norm_squared_part_2}
\end{align}

\noindent Now, similarly we get

\begin{align}
\mathbb{E}\!\left[(\hat{k}_C(\mathbf{x}, \mathbf{y}))^2\right]
&= \sum_{i=1}^d x_i^2 y_i^2 + \frac{4}{D}\sum_{i\neq i'} (x_{i}y_{i}x_{i'}y_{i'}) + \frac{D-2}{D}\sum_{i\neq i'}x_{i}y_{i}x_{i'}y_{i'}, \label{norm_squared_part_3}\\
&= \langle \mathbf{x,y}\rangle^{2} + \frac{2}{D}\sum_{i\neq i'} (x_{i}y_{i}x_{i'}y_{i'}) . \label{norm_squared_part_4}
\end{align}

\noindent Now, substitute values of $\mathbb{E}\!\left[|\hat{k}_C(\mathbf{x}, \mathbf{y})|^2\right]$ and  $\mathbb{E}\!\left[(\hat{k}_C(\mathbf{x}, \mathbf{y}))^2\right]$ and simplifying, we get

\begin{align}
\operatorname{Var}\!\big(\hat{k}(\mathbf{x}, \mathbf{y})\big)
&= \frac{1}{D}\left(\|\mathbf{x}\|_2^2 \|\mathbf{y}\|_2^2 
+ \langle \mathbf{x}, \mathbf{y} \rangle^2 
- 2 \sum_{i=1}^d x_i^2 y_i^2 \right).
\end{align}
Hence, from the above variance bound, it follows that \texttt{CountSketch} with complex random variables does not provide any variance reduction over its real-valued counterpart.
\end{proof}

\section{Further Experiments}\label{Further_exp}

\subsection{Extended Evaluation on Variance and Time in different setups} \label{Extended_experiments}

In this section, we present additional experimental results that complement the
main empirical evaluation reported in Section~\ref{Experiments_sec}. These
experiments extend the comparison between real and \emph{complex-to-real (CtR)}
sketching constructions across additional datasets, higher polynomial degrees,
and varied embedding dimension. We report both approximation quality (via KL
divergence between exact and sketched kernels) and wall-clock sketch
construction time. Together, these results provide a more detailed view of the
variance and construction time trade-offs of CtR \texttt{TensorSketch}
relative to its real-valued counterparts and JL-type baselines.

\begin{figure*}[htbp]
\centering
\includegraphics[width=\linewidth]{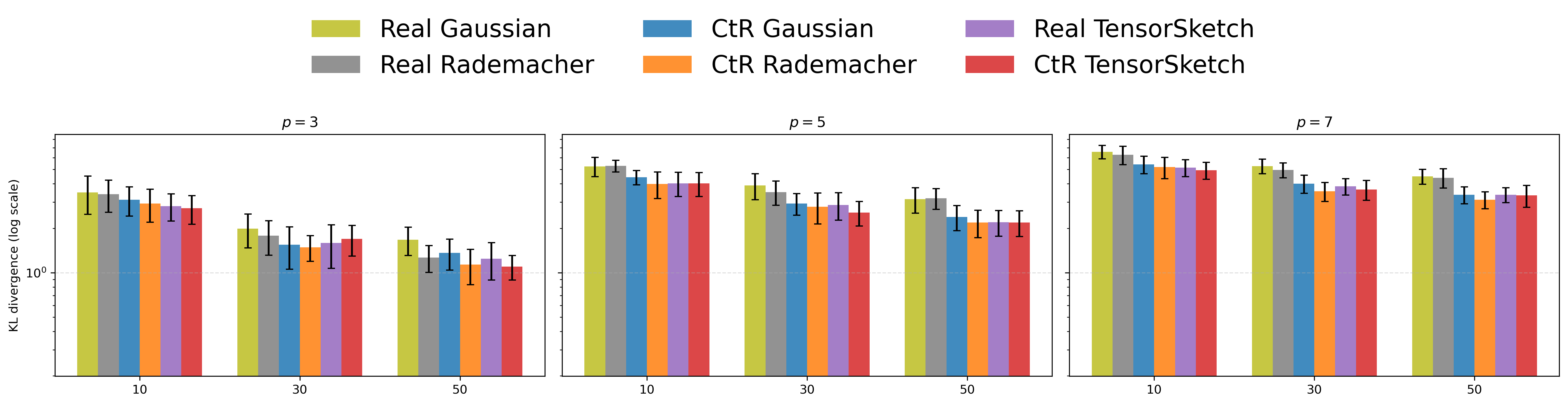}
\caption{
\textbf{KL divergence on the MAGIC Gamma Telescope dataset.}
We compare Real and \emph{complex-to-real (CtR)} Gaussian and Rademacher JL sketches,
together with Real and CtR TensorSketch.
Results are shown for polynomial degrees $p \in \{3,5,7\}$ and sketch dimensions
$D \in \{d,3d,5d\}$ with $n=3000$ standardized and $\ell_2$-normalized samples.
Bars report the mean KL divergence over 20 independent trials.
}
\label{fig:magic_kl_unoptimized_jl}
\end{figure*}

\begin{figure*}[htbp]
\centering
\includegraphics[width=\linewidth]{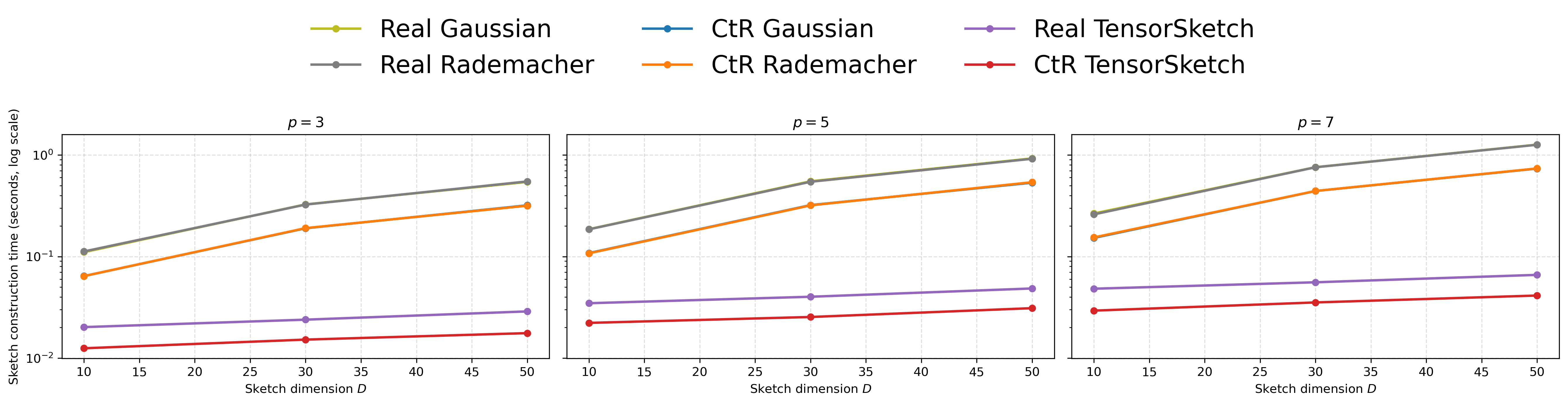}
\caption{
\textbf{Wall-clock sketch construction time on the MAGIC dataset.} Methods compared include Real and \emph{complex-to-real (CtR)} Gaussian and
Rademacher JL sketches, with Real and CtR TensorSketch.
Results are shown for polynomial degrees $p \in \{3,5,7\}$ and sketch dimensions
$D \in \{d, 3d, 5d\}$ with $n=3000$ standardized and $\ell_2$-normalized samples.
Each point reports the mean runtime over 20 independent trials for feature sketch
construction only. 
}
\label{fig:magic_unoptimized_time}
\end{figure*}

\begin{figure*}[htbp]
\centering
\includegraphics[width=\linewidth]{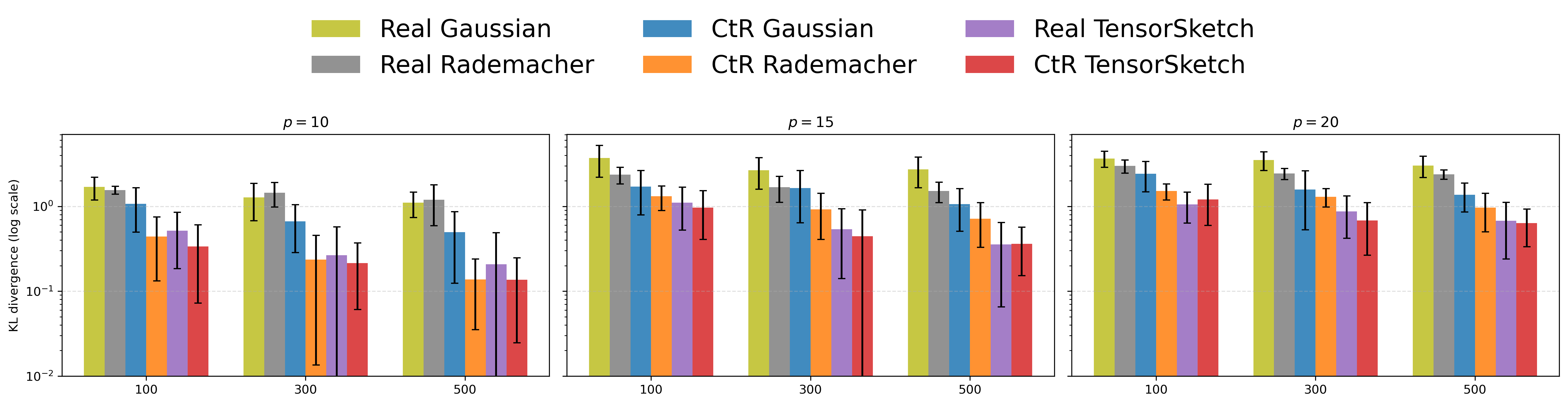}
\caption{
\textbf{KL divergence on the synthetic dataset ($d=2$).}
We compare Real and \emph{complex-to-real (CtR)} Gaussian and Rademacher JL sketches,
together with Real and CtR TensorSketch, for approximating degree-$p$
polynomial kernels on synthetic Gaussian data ($n=3000$, dimension $d=2$,
standardized and $\ell_2$-normalized).
Results are shown for polynomial degrees $p \in \{10,15,20\}$ and sketch
dimensions $D \in \{100,300,500\}$.
Bars report the mean KL divergence between the exact kernel and the
sketch-based approximation over 20 independent trials.
}
\label{fig:synth_kl_unoptimized_jl}
\end{figure*}

\begin{figure*}[htbp]
\centering
\includegraphics[width=\linewidth]{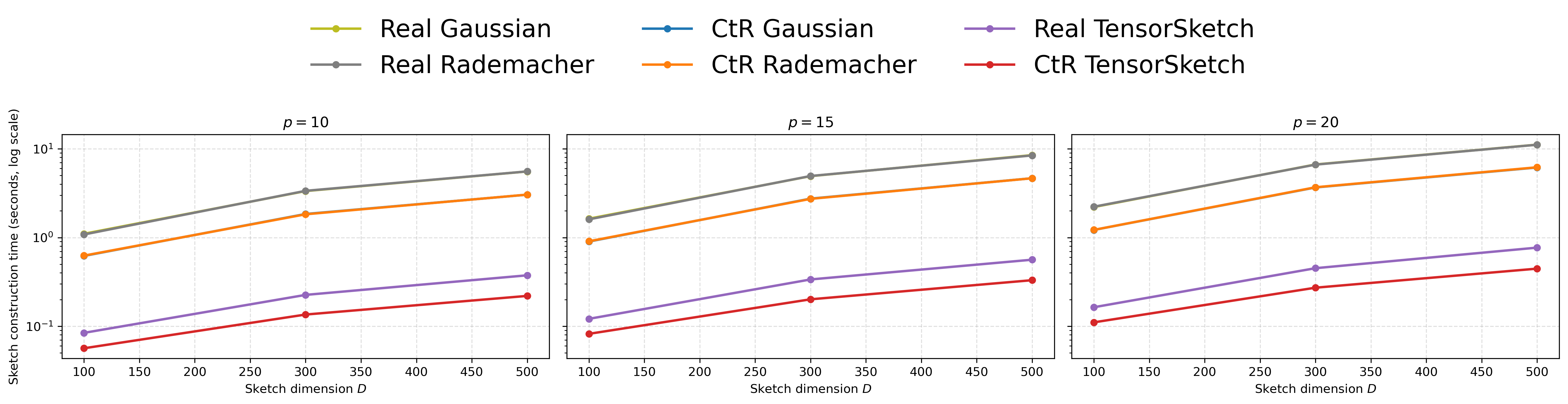}
\caption{
\textbf{Wall-clock sketch construction time on a synthetic dataset.}
Methods compared include Real and \emph{complex-to-real (CtR)} Gaussian and
Rademacher JL sketches, together with Real and CtR TensorSketch.
Results are shown for polynomial degrees $p \in \{10,15,20\}$ and sketch
dimensions $D \in \{100,300,500\}$ with $n=3000$ standardized and
$\ell_2$-normalized samples in dimension $d=2$.
Each point reports the mean runtime over 20 independent trials for
feature sketch construction only (log-scale on the $y$-axis).
}
\label{fig:synthetic_unoptimized_time}
\end{figure*}

\clearpage

\subsection{Evaluation on Frobenius Normalized Relative Error}

 We also test our proposed method along with other baselines using the Frobenius normalized relative error, which is defined as $\|K - \hat{K}\|_F / \|K\|_F$. In this formula, $K$ represents the exact kernel matrix, while $\hat{K}$ represents the approximated kernel matrix produced by the sketching method. Frobenius error metric is a standard benchmark in kernel approximation and is widely used to measure sketching accuracy \cite{wacker2022improved,pmlr-v206-wacker23a}. The datasets and experimental setups used here are identical to our previous experiments that evaluated accuracy using the KL divergence metric (Section~\ref{Experiments_sec} and Appendix~\ref{Extended_experiments}). Testing with this alternative metric ensures that our method performs well irrespective of the choice of error metric.

\begin{table*}[htbp]
\centering
\begin{tabular}{@{}llc@{}}
\toprule
\textbf{Degree ($p$)} & \textbf{Method} & \textbf{Variance} \\ \midrule
15 
& Real Gaussian & $4.41 \times 10^{4}$ \\
& Real Rademacher & $3.20 \times 10^{1}$ \\
& CtR Gaussian & $6.35 \times 10^{1}$ \\
& CtR Rademacher & $4.16 \times 10^{3}$ \\
& Real TensorSketch & $2.84 \times 10^{-5}$ \\
& \textbf{CtR TensorSketch (Ours)} & \textbf{1.57 $\times$ 10$^{-5}$} \\ \midrule

20
& Real Gaussian & $3.80 \times 10^{1}$ \\
& Real Rademacher & $6.78 \times 10^{3}$ \\
& CtR Gaussian & $2.76 \times 10^{1}$ \\
& CtR Rademacher & $1.58 \times 10^{1}$ \\
& Real TensorSketch & $1.28 \times 10^{-2}$ \\
& \textbf{CtR TensorSketch (Ours)} & \textbf{5.05 $\times$ 10$^{-5}$} \\ \midrule

25 
& Real Gaussian & $2.65 \times 10^{0}$ \\
& Real Rademacher & $9.06 \times 10^{3}$ \\
& CtR Gaussian & $1.82 \times 10^{1}$ \\
& CtR Rademacher & $2.58 \times 10^{1}$ \\
& Real TensorSketch & $3.68 \times 10^{-1}$ \\
& \textbf{CtR TensorSketch (Ours)} & \textbf{3.58 $\times$ 10$^{-4}$} \\ \midrule

30 
& Real Gaussian & $7.24 \times 10^{0}$ \\
& Real Rademacher & $1.00 \times 10^{0}$ \\
& CtR Gaussian & $5.44 \times 10^{2}$ \\
& CtR Rademacher & $3.21 \times 10^{3}$ \\
& Real TensorSketch & $1.27 \times 10^{-3}$ \\
& \textbf{CtR TensorSketch (Ours)} & \textbf{4.59 $\times$ 10$^{-6}$} \\ \bottomrule
\end{tabular}
\caption{Variance of the Frobenius normalized relative error ($\|K - \hat{K}\|_F / \|K\|_F$) evaluated on the real-world MAGIC dataset ($n=1000, d=10$) with a compressed sketch dimension of $D=128$. Results are aggregated over 20 independent trials across high-degree polynomial kernels ($p \in \{15, 20, 25, 30\}$). CtR TensorSketch strictly outperforms Real TensorSketch in estimation stability, confirming the variance reduction.}
\label{tab:frob_variance}
\end{table*}
\subsection{Evaluation on Downstream Tasks}

To demonstrate practical application beyond kernel matrix approximation, we evaluate our method on downstream binary classification tasks using a linear Support Vector Machine (SVM). The datasets used in these experiments are generated using standard ML libraries (such as scikit-learn) to create two highly overlapping classes with non-linear decision boundaries. We test two separate configurations under extreme compression to a sketch dimension of $D=64$.

We compare the sketching methods against the Exact Polynomial Kernel, which computes the full kernel matrix using the mathematical formula $K(x, y) = (x^\top y)^p$ without any compression or approximation. We evaluate the techniques using two metrics: test classification accuracy and total execution time. The time metric tracks the combined end-to-end wall-clock time required for both data compression and the subsequent classifier training. We report this total time to fully reflect the complete workload required to obtain the final classification model.

As shown in Table~\ref{tab:downstream_setup1} and Table~\ref{tab:downstream_setup2}, all approximation methods experience a drop in accuracy compared to the exact kernel due to the tight bottleneck of the compressed dimension. However, our proposed CtR \texttt{TensorSketch} consistently achieves the highest accuracy among all baselines while requiring the shortest total execution time.

\begin{table}[htbp]
\centering
\begin{tabular}{@{}lcc@{}}
\toprule
\textbf{Method} & \textbf{Accuracy} & \textbf{Time (s)} \\ \midrule
Exact Poly Kernel & 0.8756 & 16.0724 \\ 
TensorSketch (Real) & 0.4911 & 0.2390 \\
\textbf{TensorSketch (CtR) [Ours]} & \textbf{0.5289} & \textbf{0.1479} \\
JL (CtR Rademacher) & 0.5011 & 3.8279 \\
JL (CtR Gaussian) & 0.4978 & 3.7708 \\ \bottomrule
\end{tabular}
\caption{Downstream Classification Performance (Linear SVM) for \textbf{Setup 1}. Evaluated on a synthetic dataset ($n=3000$) with dimension $d=20$. We approximate a polynomial kernel of degree $p=15$ using a sketch dimension of $D=64$. Our proposed method achieves the best sketching accuracy in the shortest time.}
\label{tab:downstream_setup1}
\end{table}

\begin{table}[htbp]
\centering
\begin{tabular}{@{}lcc@{}}
\toprule
\textbf{Method} & \textbf{Accuracy} & \textbf{Time (s)} \\ \midrule
Exact Poly Kernel & 0.8567 & 8.5344 \\ 
TensorSketch (Real) & 0.4978 & 0.2763 \\
\textbf{TensorSketch (CtR) [Ours]} & \textbf{0.5044} & \textbf{0.1733} \\
JL (CtR Rademacher) & 0.4567 & 3.3204 \\
JL (CtR Gaussian) & 0.4889 & 3.2954 \\ \bottomrule
\end{tabular}
\caption{Downstream Classification Performance (Linear SVM) for \textbf{Setup 2}. Evaluated on a synthetic dataset ($n=3000$) with dimension $d=10$. We approximate a polynomial kernel of higher degree $p=25$ using a sketch dimension of $D=64$.}
\label{tab:downstream_setup2}
\end{table}

\subsection{Variance Analysis through Numerical Table}

The primary purpose of this section is to provide a clear and direct validation of our experimental findings from Section~\ref{Experiments_sec}. In our main evaluation, we rely on visual plots to demonstrate the accuracy and computational time of our proposed CtR \texttt{TensorSketch}. Graphs, especially those that use a logarithmic scale, can visually compress the performance gap between methods. Looking at the exact numbers, we can clearly confirm our previous experiments and highlight the advantage of our method. Specifically, these numbers reveal how our approach successfully reduces the variance from $3^p/D$ in traditional methods to $2^p/D$ in our complex-to-real design. 

As shown in Table~\ref{tab:variance_kl}, our proposed CtR \texttt{TensorSketch} consistently achieves lower variance compared to all the baselines across all degrees.
\begin{table*}[htbp]
\centering
\begin{tabular}{@{}llc@{}}
\toprule
\textbf{Degree ($p$)} & \textbf{Method} & \textbf{Variance} \\ \midrule
15 
& Real Gaussian & 0.2634 \\
& Real Rademacher & 0.3839 \\
& CtR Gaussian & 0.2192 \\
& CtR Rademacher & 0.2029 \\
& Real TensorSketch & 0.5496 \\
& \textbf{CtR TensorSketch (Ours)} & \textbf{0.1739} \\ \midrule

20
& Real Gaussian & 0.3545 \\
& Real Rademacher & 0.4673 \\
& CtR Gaussian & 0.6996 \\
& CtR Rademacher & 0.5246 \\
& Real TensorSketch & 0.3932 \\
& \textbf{CtR TensorSketch (Ours)} & \textbf{0.3888} \\ \midrule

25
& Real Gaussian & 0.3158 \\
& Real Rademacher & 0.2560 \\
& CtR Gaussian & 0.4395 \\
& CtR Rademacher & 0.2294 \\
& Real TensorSketch & 0.3209 \\
& \textbf{CtR TensorSketch (Ours)} & \textbf{0.2367} \\ \midrule

30
& Real Gaussian & 0.4785 \\
& Real Rademacher & 0.2932 \\
& CtR Gaussian & 0.3266 \\
& CtR Rademacher & 0.2883 \\
& Real TensorSketch & 0.2443 \\
& \textbf{CtR TensorSketch (Ours)} & \textbf{0.1939} \\ \bottomrule
\end{tabular}
\caption{Variance of KL Divergence for high-degree polynomial kernels. Evaluated on a positive orthant synthetic dataset ($n=1000, d=10$) with extreme compression ($D=64$), aggregated over 20 independent trials. CtR TensorSketch strictly dominates Real TensorSketch in estimation variance across all polynomial degrees, empirically validating the theoretical $2^p/D$ scaling advantage.}
\label{tab:variance_kl}
\end{table*}

\end{document}